\documentclass[a4paper,UKenglish,cleveref, autoref, thm-restate]{lipics-v2021}
 
\pdfoutput=1
\hideLIPIcs
 
\usepackage{graphicx}
\usepackage{textcomp}
\usepackage{xcolor}
\usepackage{url}
\usepackage[ruled,vlined,linesnumbered]{algorithm2e}
\DontPrintSemicolon
\SetKw{Return}{return}
\SetKw{Break}{break}
\usepackage{etoolbox}
\usepackage{cleveref}
\usepackage{multirow}
\usepackage{graphicx}
\usepackage{siunitx}
\usepackage{bm}
\usepackage{svg}
\definecolor{darkgray}{gray}{0.4}

\SetCommentSty{graycommentstyle}
\usepackage{stmaryrd}
\newcommand\Tau{\mathcal{T}}
\usepackage{subcaption}
\usepackage{eurosym}
\newtheorem{property}{\textbf{Property}}
 
\newcommand{\cbribe}{c_{\mathrm{bribe}}}      
\newcommand{\Rmev}[1]{\mathcal{R}(#1)}        
\newcommand{\Rpub}[1]{\mathcal{R}_{\mathrm{pub}}(#1)} 
\newcommand{\Bbase}{\mathcal{C}^{\mathrm{Eth}}} 
\newcommand{\Glimit}{g_{\max}}                 
\newcommand{\ftipgas}[1]{f^{#1}_{\mathrm{tip,\,gas}}}
\newcommand{\ftiptx}[1]{f^{#1}_{\mathrm{tip,\,tx}}}
\newcommand{\fbasegas}[1]{f^{#1}_{\mathrm{base,\,gas}}}
\newcommand{\fbasetx}[1]{f^{#1}_{\mathrm{base,\,tx}}}
\newcommand{\pool}{\mathrm{Pool}}
\newcommand{\extpool}{\mathrm{ExtendedPool}}
 
\title{Accountable Transaction Inclusion Lists:\, Enhancing Ethereum's
Censorship Resistance}
 
\titlerunning{Accountable Transaction Inclusion Lists} 
 
\author{Patrick Spiesberger}{Karlsruhe Institute of Technology, Germany}{patrick.spiesberger@kit.edu}{https://orcid.org/0009-0006-4746-8868}{}
\author{Hannes Hartenstein}{Karlsruhe Institute of Technology, Germany}{hannes.hartenstein@kit.edu}{https://orcid.org/0000-0003-3441-3180}{}
 
\authorrunning{P. Spiesberger and H. Hartenstein}
 
\Copyright{Patrick Spiesberger and Hannes Hartenstein} 

 

\nolinenumbers 
 
\acknowledgements{This work was funded by the \href{https://kastel-labs.de/}{KASTEL Security Research Labs.}}

\begin{document}
 
\maketitle

\begin{abstract}
In Ethereum, transaction inclusion is rarely in question; what matters is the delay until inclusion.
Currently, block builders could exercise censorship across consecutive blocks, threatening time-critical applications, such as on-chain auctions.
To mitigate this risk, existing proposals such as \textsc{FOCIL}, scheduled for deployment in late 2026, assign a committee to list transactions for mandatory inclusion.
However, no committee member is held accountable for the actual inclusion of the transactions: an adversary can bribe the entire committee to omit any transaction for less than $2~\text{\euro}$ per block under current conditions.
We argue that accountability, i.e., requiring all exclusion decisions to be publicly disclosed and verifiably complete, with violations attributable to a specific party, substantially raises censorship costs.
To this end, we propose \textsc{Fair Forward Inclusion Lists} (\textsc{FairFIL}) as an accountable censorship resistance mechanism for Ethereum.
In \textsc{FairFIL}, every builder must publish all transactions the builder chooses to censor, subject to a protocol-anchored policy; a committee verifies the completeness and validity of this disclosure.
The subsequent builder must include these transactions, forfeiting the full block reward upon any omission.
Therefore, under \textsc{FairFIL}, extending censorship beyond a single slot requires an assembler to forfeit a full block reward.
We show that compliance is rational for all participants within our behavior model.
Our empirical evaluation on Ethereum mainnet indicates that multi-block censorship costs one order of magnitude more than under existing proposals, while leaving the builder's MEV extraction freedom largely intact.
Initial measurements further suggest that the mempool consistency \textsc{FairFIL} requires is met in practice.
\end{abstract}

\section{Introduction}
\label{sec:introduction}
Ethereum~\cite{buterin_ethereum_2014}, as of mid 2026, is the largest general-purpose smart-contract platform by market capitalization, hosting a wide range of applications that depend on transaction inclusion with low and predictable latency.
Ethereum operates on a time-slotted model~\cite{ETHspec_consensus}: in each \SI{12}{\second} slot, a single leader (in Ethereum terminology the \emph{proposer}) is selected at random to assemble a block, while the remaining participants (\emph{validators}) attest to its validity.
During the slot, the proposer has short-lived but unilateral control over which transactions are included in the block and in what order~\cite{gramlich_mev_2024,schneider1990}.
Correspondingly, any transaction can be excluded at will by the proposer.
For some transactions, inclusion within a handful of slots is sufficient.
However, a broad class of transactions is time-critical: each additional slot of inclusion latency can diminish the value delivered to the user or application, or more broadly undermine Ethereum's attractiveness as a reliable platform.
Examples include auctions, decentralized finance, and commit-and-reveal protocols. In auctions, a bid submitted shortly before the auction closes may be suppressed after the deadline~\cite{fox_censorship_2023,wu2024strategicbiddingwarsonchain}. In decentralized finance, a stop-loss transaction submitted when a trigger price is reached can be censored; each additional delay allows the asset price to move further against the user's intent~\cite{daian_flash_2020,NADLER2026102908}.
\newpage
\noindent In commit-and-reveal schemes, a reveal transaction must be submitted within a bounded window after committing; if censored, the protocol cannot distinguish the outcome from an intentional non-reveal, and the participant may be penalized despite having acted in time.

For such applications, the essential question is not whether a transaction will eventually be included in the blockchain, but how quickly the transaction affects the application.
We follow the established usage in the Ethereum literature~\cite{buterinProblem2015,wahrstatter_blockchain_2024} and use the term ``censorship'' for additional inclusion latency --~the exclusion of a transaction from one or more consecutive blocks -- rather than for permanent exclusion from the blockchain~\cite{fox_censorship_2023}.
As long as control over inclusion decisions rotates frequently among independent parties, no single party can delay a transaction beyond a limited timeframe.
Concerns arise when one party controls block construction over extended periods, a scenario facilitated by Ethereum's Proposer--Builder Separation (PBS)~\cite{heimbach_ethereums_2023,EIP7732}.
Under PBS, proposers delegate block construction to specialized builders to capture Maximal Extractable Value (MEV), i.e., the additional revenue obtained by strategically including, excluding, or reordering transactions within or across blocks~\cite{gramlich_mev_2024}.
Empirical evidence shows that two builders assembled more than \SI{80}{\percent} of all blocks in 2024 and 2025~\cite{yang_decentralization_2025,bahrani_centralization_2025}, and nearly all proposers rely on them~\cite{ethereum_altruistic_proposers_2026}.
A dominant builder could, in principle, impose censorship over many consecutive slots, thereby undermining the reliability of the network.
Our measurements, as described below, underline the relevance in practice: in September 2025, an average of \num{8.3} transactions per block were censored despite being eligible for inclusion under the intended tip-based reference ordering~\cite{EIP-1559}.

To mitigate transaction censorship, prior work proposes mechanisms that distribute inclusion control across a committee of $\kappa$~participants per slot~\cite{EIP7547,thiery2024eip7805,aucil,garimidi_concurrent_2025}.
The common idea is to split one party's exclusive control into shares across $\kappa$ parties, such that a single honest member suffices to force the inclusion of a transaction.
However, as we show later, an adversary can bribe the entire committee of any of these mechanisms at a cost of only a few euros per slot under current conditions.
We attribute this low cost of censorship to the fact that no individual committee member is held accountable: no member's exclusions are publicly attributable, verifiable, or subject to protocol-enforced penalties.
In this paper, we therefore address the following research question:

\vspace{-2mm}
\subparagraph{Research question.}
\emph{Does accountability for exclusion decisions substantially increase the cost of censoring a transaction?}

\vspace{-2mm}
\subparagraph{Contribution and Rationale.}
The idea of accountability in censorship-resistance mechanisms and an initial sketch of \textsc{FairFIL} were first introduced in~\cite{spiesberger_policy-based_2025}. This paper builds on that foundation and makes three main contributions.
First, we formalize accountability for censorship resistance in Ethereum.
We prove that accountability can substantially increase the cost of sustained censorship: under accountable inclusion enforcement, a censoring builder forfeits the entire block reward, increasing censorship costs by roughly one order of magnitude compared to prior approaches.
Second, we present \textsc{Fair Forward Inclusion Lists}~(\textsc{FairFIL}), an accountable censorship resistance mechanism, and prove that under a fully rational participant model, \textsc{FairFIL} enforces block invalidation under sustained censorship.
Third, we implement \textsc{FairFIL} and empirically evaluate the mechanism on Ethereum mainnet, quantify the cost of censorship and provide first evidence for practical deployability.
\textsc{FairFIL} is designed to be explicitly MEV-friendly, preserving most of the builders' flexibility for MEV extraction.
Two considerations motivate this choice of an ``MEV-friendly design''.
First, MEV constitutes a substantial share of proposer revenue~\cite{MEV_Attacks}. 
Curtailing MEV-extraction freedom risks undermining Ethereum's long-term viability. 
Though difficult to quantify, we regard this risk as significant, as does prior work~\cite{mev_sharing}.
Second, over \SI{90}{\percent} of validators~\cite{heimbach_ethereums_2023,ethereum_altruistic_proposers_2026} voluntarily adopted MEV-Boost~\cite{flashbots_mev-boost}, 
despite censorship and centralization risks~\cite{heimbach_ethereums_2023,yang_decentralization_2025}, demonstrating that validators will deviate from protocol intentions when compliance reduces revenue.
\textsc{FairFIL} trades only a small share of the block assembler's MEV-extraction freedom in exchange for inclusion guarantees: any transaction may be censored for one slot but must be included thereafter, keeping assembler constraints low while guaranteeing timely inclusion for the transaction sender.

\vspace{-2mm}
\subparagraph{Paper structure.}
\Cref{sec:background} introduces the system model, defines transaction censorship, and states the participant and adversary model.
\Cref{sec:relatedwork} surveys prior censorship resistance mechanisms and identifies the absence of accountability as their main limitation.
\Cref{sec:fairfil} presents the \textsc{FairFIL} protocol.
\Cref{sec:properties} empirically evaluates \textsc{FairFIL} on Ethereum mainnet, formalizes accountability as four properties satisfied by \textsc{FairFIL}, and quantifies the per-slot cost of sustained censorship compared to prior mechanisms.
\Cref{sec:discussion} discusses assumptions and deployment considerations regarding MEV; \Cref{sec:conclusion} concludes.

\section{Foundations and Design Space}
\label{sec:background}

\subsection{System Model and Transaction Censorship}
\label{sec:model}
 
We consider an abstracted Ethereum model in which the blockchain is extended in discrete time slots.
For each slot~$s$, a committee $V_s \subset V$ of validators is pseudo-randomly selected from the active validator set~$V$, while a single proposer $v_s^{P} \in V_s$ is responsible for assembling the block at the beginning of slot~$s$~\cite{ETHspec_consensus}.
We refer to the party that actually constructs the block -- whether the proposer itself or a delegated builder under PBS~\cite{EIP7732} -- as the \emph{block assembler}~$\beta_s$.
\Cref{tab:notation} summarizes the symbols and abbreviations used throughout this paper.

\vspace{-2mm}
\subparagraph{Transactions.}
A transaction~$\tau$ is a signed user statement~\cite{ETHspec_consensus} specifying an intent and an execution cost $\mathrm{gas}(\tau)$ determined upon execution.
For every unit of gas consumed, the user pays a voluntary per-gas tip~$\ftipgas{\tau}$ -- intended to incentivize faster inclusion -- to the block assembler~$\beta_s$, and a mandatory per-gas base fee~$\fbasegas{\tau}$ that is received by no party~\cite{EIP-1559}.
Accordingly, the total per-transaction tip paid to block assembler~$\beta_s$ for including~$\tau$ is $\ftiptx{\tau} := \ftipgas{\tau} \cdot \mathrm{gas}(\tau)$, and the total per-transaction base fee is $\fbasetx{\tau} := \fbasegas{\tau} \cdot \mathrm{gas}(\tau)$.

\vspace{-2mm}
\subparagraph{Mempool.}
Transactions are disseminated through Ethereum's peer-to-peer~(P2P) layer and stored, at each node, in a local buffer called the mempool.
We write $\pool_s^{n}$ for the mempool view of a node~$n$ at slot~$s$; in particular, $\pool_s^{\beta_s}$ is the block assembler's view.
For the theoretical analysis, we assume all validators share an identical mempool view, consistent with prior analyses of censorship resistance mechanism designs~\cite{aucil}.
In practice, mempools are not perfectly consistent; we revisit this assumption empirically in \Cref{sec:empirical:mempool} and find that the near-perfect consistency observed appears to be sufficient for \textsc{FairFIL} operation.

\vspace{-2mm}
\subparagraph{Block construction and validity.}
At the beginning of slot~$s$, the block assembler selects a set of transactions from $\pool_s^{\beta_s}$, possibly augmented with transactions received via a private communication channel (Exclusive Order Flow, XOF)~\cite{gramlich_mev_2024}.
Subsequently, the assembler orders and executes these transactions and publishes the resulting block~$B_s$.
A block is protocol-valid if the execution on top of the state induced by block $B_{s-1}$ respects Ethereum's execution rules~\cite{wood_ethereum_2025}; in particular, the cumulative gas consumed by $B_s$ must not exceed the per-slot gas limit~$\Glimit(B_s)$.
A block is accepted once a majority of $V_s$ attests to its validity; due to Ethereum's proposer boost, slightly less than a majority suffices in practice~\cite{schwarzschilling_proposer_boost}.

\vspace{-2mm}
\subparagraph{Transaction Censorship.}
\label{sec:defcensorship}
In this work, we restrict the notion of censorship to a specific operational instance: the exclusion of a transaction from the block currently being assembled.
Censorship in this sense denotes additional inclusion latency, not permanent exclusion from the blockchain.
We neither consider self-censorship (users refraining from submitting transactions) nor censorship at Ethereum's P2P layer.
To decide which transactions should have been included, we adopt a deterministic reference ordering.
Transactions in the mempool $\pool_s^{\beta_s}$ are sorted in descending order of their per-gas tip $\ftipgas{\tau}$, with the transaction hash as a tie-breaker.
Iterating through this sequence, we admit every transaction whose execution preserves block validity; we call the resulting sequence the \emph{Fair Ordering}\footnote{This use of ``fair'' differs from receive-time-based notions of order fairness such as the Aequitas family~\cite{kelkar_orderfairness_2020}.} of the mempool.

\begin{definition}[Transaction censorship]
\label{def:censorship}
A transaction $\tau$ is \emph{censored} from block $B_s$ if $\tau$ was propagated before the start of slot~$s$ via Ethereum's P2P layer and its inclusion in $B_s$ under the Fair Ordering would preserve block validity, yet transaction $\tau$ is excluded from block $B_s$.
\end{definition}

The condition ``propagated before the start of slot~$s$'' is not precisely observable in a decentralized system. 
In practice, we operationalize this condition as the initial propagation at least \SI{3}{\second} before the start of the slot to account for network propagation delays.
\Cref{tab:censorship-example} illustrates \Cref{def:censorship} on a mempool $\pool_s^{\beta_s} = \{\tau_1,\dots,\tau_5\}$, augmented with a privately submitted transaction $\tau_{\mathrm{private}}$, against a block gas limit of $\Glimit = 100$ gas.
The candidate sequence obtained by sorting $\pool_s$ by per-gas tip in descending order is $[\tau_1, \tau_2, \tau_3, \tau_4, \tau_5]$.
Iterating through this sequence and admitting each transaction whose execution still fits into $\Glimit$ yields the Fair Ordering $[\tau_1, \tau_2, \tau_4, \tau_5]$ with cumulative gas~\num{50}.
Transaction $\tau_3$ is dropped because including $\{\tau_1, \tau_2, \tau_3\}$ would exceed $\Glimit$.
By \Cref{def:censorship}, any transaction in the \textit{Fair Ordering} not included in $B_s$ is censored.
In our example, this applies to $\tau_2$, as well as to $\tau_4$, which is displaced by the privately submitted $\tau_{\mathrm{private}}$.

\begin{table}[!ht]
\centering
\begin{tabular}{c|c|c|c|c|c}
$\tau$ & Tip $\ftipgas{\tau}$ & $\mathrm{gas}(\tau)$ & Includable? & $\in B_s$? & Censored? \\
\hline
$\tau_5$ & 2 & 10 & Yes & $\checkmark$ & No \\
$\tau_2$ & 8 & 10 & Yes & $\times$ & \textbf{Yes} \\
$\tau_3$ & 7 & 90 & No & $\times$ & No \\
$\tau_{\mathrm{private}}$ & 20 & 70 & Yes & $\checkmark$ & -- \\
$\tau_4$ & 5 & 10 & Yes & $\times$ & \textbf{Yes} \\
$\tau_1$ & 10 & 20 & Yes & $\checkmark$ & No \\
\end{tabular}
\vspace{1mm}
\caption{Set of transactions available to the current block assembler for block construction ($\Glimit = 100$), resulting in block $B_s = [\tau_5, \tau_{\mathrm{private}}, \tau_1]$.  
While this selection may appear arbitrary, it reflects the assembler's MEV extraction strategy, which encompasses both the censorship of transactions and the inclusion of private transactions, with the objective of maximizing revenue.
}
\label{tab:censorship-example}
\end{table}

\vspace{-4mm}
\noindent We further classify censored transactions by how long their inclusion is delayed.
 
\begin{definition}[$n$-slot censorship]
\label{def:nslotcensorship}
Let $s$ be the earliest slot in which transaction $\tau$ is includable.
For $n \in \mathbb{N}_+$, $\tau$ experiences \emph{$n$-slot censorship} if it is censored from blocks $B_s, \ldots, B_{s+n-1}$ and included in block $B_{s+n}$.
The case $n = 0$ corresponds to the absence of censorship for $\tau$.
\end{definition}

\subsection{Behavior Model and Censorship Resistance Mechanisms}
\label{sec:cr_framework}
We apply the Byzantine, Altruistic, Rational (BAR) model by Aiyer et al.~\cite{aiyer_bar_2005} to characterize validator and block assembler behavior, following prior work such as~\cite{ethereum_altruistic_proposers_2026}.
An \textit{altruistic} participant follows the protocol's intent regardless of payoff.
We assume their absence, as protocols should be robust even in the worst plausible case, and revisit this assumption in \Cref{discussion:altruism}.
In what follows, all validators and block assemblers are modeled as \textit{rational participants}: a rational participant deviates from the protocol whenever deviation yields a strictly higher personal payoff than compliance, and may cooperate with other rational participants.
Such a deviation need not violate the protocol; it may instead exploit degrees of freedom the protocol permits, contrary to its intent.
We treat censorship as a deviation from the Ethereum protocol~\cite{ethereum-vision} and later extend this to deviations from censorship-resistance protocols.
When compliance and deviation yield identical payoffs, we assume compliance: deviating offers no financial gain, while `unfair' behavior could undermine Ethereum's long-term viability and thereby reduce future validator revenue.
However, we assume that any strictly positive revenue gain suffices to induce deviation.
Adversary~$\mathcal{A}$ is not profit-driven and aims to censor transactions: $\mathcal{A}$ can communicate directly with all participants~\cite{aucil} and may bribe rational participants to deviate, subject to a per-slot bribing budget~$\mathcal{M}$.
All bribes are paid on-chain via a bribing smart contract~\cite{sookitoth_bribers_2025}; out-of-band payments (e.g., bank transfers) are outside our model.
Each bribe therefore incurs a fixed overhead~$\cbribe$ for adversary~$\mathcal{A}$, covering fund transfer and on-chain verification -- in addition to the bribe amount itself.

\subparagraph{Budget restriction.}
In Ethereum, censorship resistance is fundamentally an economic property~\cite{fox_censorship_2023}. 
As long as block assemblers may choose to withhold a block, any sufficiently wealthy adversary can censor transactions by compensating the block assembler~$\beta_s$ for the forgone block reward~$\Rmev{B_s}$, defined as the sum of transaction tips and extracted MEV.
Providing censorship resistance against arbitrarily wealthy adversaries would require substantial structural changes to Ethereum and lies outside the scope of this paper.
Accordingly, we assume that an adversary cannot bribe a block assembler to an extent that prevents the publication of a block.
Under this assumption, the adversary's per-slot budget is bounded by the block reward~$\Rmev{B_s}$, since any larger budget would enable the adversary to prevent block publication in every slot, thereby violating Ethereum's liveness property and halting chain progression.
We further assume that the adversary cannot bribe the slot committee~$V_s$ to accept an invalid block or reject a protocol-valid one.
Such attacks would require bribing a majority of the committee, which comprises approximately \num{30000} validators in 2025~\cite{heimbach_deanonymizing_2025}.
We argue that the required budget significantly exceeds the block reward~$\Rmev{B_s}$.

\vspace{2mm}
\noindent Based on this adversary model, we define censorship resistance as follows:

\begin{definition}[$(\mathcal{M},\sigma)$-censorship resistance]
\label{def:cr}
A protocol is \emph{$(\mathcal{M},\sigma)$-censorship-resistant} if, for any transaction~$\tau$, an adversary with per-slot budget~$\mathcal{M}$ cannot prevent $\tau$ from being included for more than $\sigma$ consecutive slots under rational behavior of all other participants.
\end{definition}
 
We write $\mathcal{C}_{\tau} = [\mathcal{C}_{\tau}^{s}, \mathcal{C}_{\tau}^{s+1}, \dots]$ for the per-slot cost sequence incurred by $\mathcal{A}$ to censor $\tau$, where each entry $\mathcal{C}_{\tau}^{s+i}$ denotes the cost in the respective slot $s+i$.
The summation of costs up to the $n$-th element therefore corresponds to the cost of $n$-slot censorship of $\tau$.

\subparagraph{Ethereum baseline.}
We establish the censorship resistance of plain Ethereum as of early 2026.
Informally, outbidding transaction $\tau$'s per-transaction tip $\ftiptx{\tau}$ by any $\epsilon > 0$ makes accepting the bribe more profitable than including $\tau$, thereby censoring $\tau$ for one slot; the attack can be repeated in each slot.
This yields \Cref{lemma:EthereumCR}; the proof is given in Appendix \ref{appendix:proof_ethereum}.

\newpage
\begin{lemma}[Ethereum baseline]
\label{lemma:EthereumCR}
Ethereum is $\bigl(\ftiptx{\tau} + \cbribe,\; 0\bigr)$-censorship-resistant.
The per-slot cost of excluding transaction $\tau$ is $\Bbase_{\tau} := \ftiptx{\tau} + \cbribe + \epsilon$ for each slot, with $\epsilon > 0$.
\end{lemma}
In absolute terms, a budget of $0.10~\text{\euro}$ per slot is sufficient to censor the median transaction, where transactions are ranked by total tips, observed between September and December 2025. 
The required budget is predominantly driven by the bribe overhead~$\cbribe$.

While the per-slot cost of censorship is financially bounded by the block reward, the objective of a censorship resistance mechanism is to maximize the cost an adversary must incur to censor a target transaction.
We consider inclusion list mechanisms that achieve this by designating a set of required transactions~$\Tau_{\mathrm{req}}$ whose inclusion is protocol-enforced.
Accordingly, we regard the transaction set $\Tau_{\mathrm{req}}$ as an obligation imposed on a specified party.

\begin{definition}[Inclusion List-Based Censorship Resistance Mechanism, IL-CR]
\label{def:cr_mechanism}
An inclusion list-based censorship resistance mechanism~$\Pi_n$ designates, for each slot~$s$, a set~$\Tau_{\mathrm{req}}$ of transactions and requires that every transaction in $\Tau_{\mathrm{req}}$ is included no later than block~$B_{s+n}$.
Such a mechanism is called accountable if the construction of $\Tau_{\mathrm{req}}$ and its enforcement are publicly verifiable, so that any deviation can be detected and penalized.
\end{definition}

\noindent We distinguish two classes of IL-CR mechanisms: a \emph{constructing} mechanism assigns construction of $\Tau_{\mathrm{req}}$ to a committee without content constraints, and is therefore non-accountable.
In contrast, a \emph{verifying} mechanism defines a protocol-anchored construction policy against which a committee verifies compliance -- any deviation is detectable, regardless of who performed the construction.
A verifying mechanism is \emph{accountable} if every policy violation can be attributed to a specific party and the protocol specifies a corresponding consequence.

The significance of this distinction follows from the rationality model.
In constructing mechanisms, each committee member independently contributes a partial list -- a subset of transactions selected at the member's discretion -- from which $\Tau_{\mathrm{req}}$ is assembled as the aggregate.
No member faces any penalty for omitting transaction~$\tau$ from their partial list; any bribe exceeding the per-transaction tip therefore suffices to induce omission.
Further, without a policy constraint, members may insert privately negotiated transactions, which can displace legitimate mempool transactions, enable advanced MEV strategies, and undermine external verifiability.
Under the rationality assumption, this is the predicted outcome in practice: an outsourcing of $\Tau_{\mathrm{req}}$ construction to specialized builders, analogous to PBS, thus becomes plausible.
In verifying mechanisms with accountability, by contrast, any deviation exposes the constructing party to a financial penalty independent of transaction fees -- such as forfeiting the full block reward~$\Rmev{B_s}$ -- making omission far more costly than any per-transaction tip.
We demonstrate in this paper that \textsc{FairFIL}, as an accountable verifying mechanism, raises the per-slot cost of sustained censorship by roughly one order of magnitude relative to the constructing mechanisms surveyed in the following \Cref{sec:relatedwork}.

\section{Related Work}
\label{sec:relatedwork}
We consider four prior censorship resistance mechanism proposals: FIL~\cite{EIP7547}, FOCIL~\cite{thiery2024eip7805}, MCP~\cite{garimidi_concurrent_2025}, and AUCIL~\cite{aucil}, where FIL, FOCIL, and AUCIL are inclusion list-based censorship resistance mechanisms. 
We further classify MCP as an IL-CR mechanism despite structural differences, which we discuss in detail below.
The FIL approach represents the most basic instance of an IL-CR mechanism, where $\Tau_{\mathrm{req}}$ is constructed by a single party ($\kappa = 1$). 
In contrast, the other mechanisms distribute the construction of $\Tau_{\mathrm{req}}$ across a committee of size $\kappa > 1$ per slot.
For each mechanism, we describe the design and derive the resulting per-slot censorship cost, as detailed in Appendix~\ref{appendix:proofs_priorcr}. 
To the best of our knowledge, no existing mechanism provides accountability for transaction exclusion. 

\vspace{4mm}
Forward Inclusion Lists (FIL)~\cite{EIP7547} (Neuder et al.) is a $\Pi_1$ mechanism that enforces transaction inclusion in a \emph{forward} manner: the inclusion list ($\Tau_{\mathrm{req}}$) is constructed by the current block assembler and published alongside block $B_s$, while binding obligations apply to the subsequent block assembler $\beta_{s+1}$, who must include all listed transactions in block~$B_{s+1}$.
Since FIL does not constrain the construction of $\Tau_{\mathrm{req}}$, the assembler of $B_s$ retains full freedom over its contents, and omissions are not verifiable. Tsao et al.~\cite{Tsao2023FIL} introduce a first step towards accountability by adding a validity check that ensures listed transactions are includable in the next block, but the construction of obligation $\Tau_{\mathrm{req}}$ remains unconstrained.
Since the block assembler receives no reward tied to any specific transaction in $\Tau_{\mathrm{req}}$, bribing the assembler to omit a transaction from the list incurs only a negligible additional cost. 
The dominant censorship cost is that of excluding transaction~$\tau$ from both block~$B_s$ and $B_{s+1}$ (see \Cref{lemma:EthereumCR}).
A tighter characterization of per-slot censorship costs under FIL is provided in Appendix~\ref{appendix:proof_fil}.
\textsc{FairFIL} strengthens FIL's design by introducing accountability through a protocol-defined and verifiable construction of $\Tau_{\mathrm{req}}$, while leveraging the block reward as an intrinsic enforcement mechanism.

Fork-Choice Enforced Inclusion Lists (FOCIL)~\cite{thiery2024eip7805} (Thiery et al.), scheduled for deployment on Ethereum in 2026~\cite{ethereum_hegota_2026}, distributes the construction of $\Tau_{\mathrm{req}}$ across a committee of $\kappa = 16$ validators per slot.
Each committee member independently selects transactions from its local mempool view and broadcasts a partial inclusion list during slot $s$. 
The block assembler $\beta_{s+1}$ is required to include all transactions appearing in any partial list, subject to block validity and capacity constraints.
In contrast to FIL, transactions propagated within slot~$s$ (in particular up to three seconds before the start of slot~$s{+}1$) are subject to an inclusion obligation for block assembler $\beta_{s+1}$, which classifies FOCIL as a $\Pi_0$ mechanism.
FOCIL relies on a \emph{one-of-$\kappa$-honest} assumption: a single honest committee member suffices to enforce inclusion of a transaction $\tau$ in $\Tau_{\mathrm{req}}$, implying that an adversary must compromise all $\kappa$ members to censor $\tau$.
However, as in FIL, committee members receive no reward for specific transactions they include; any positive bribe is therefore sufficient to induce omission from an individual member's partial list.
Consequently, censorship resistance scales with committee size but remains vulnerable in a purely rational setting. 
We show in Appendix~\ref{appendix:proof_focil} that the resulting per-slot bribery cost across the entire committee is below \num{2}~\euro{} for a median-tip transaction in late~2025 under current Ethereum conditions (see \Cref{sec:empirical_baseline}).

Multiple Concurrent Proposers (MCP)~\cite{garimidi_concurrent_2025} (Garimidi et al.) departs from the inclusion-list paradigm of FIL and FOCIL by introducing $\kappa$ concurrent block assemblers (we assume $\kappa = 16$), each constructing an independent partial block.
These $\kappa$ partial blocks are concatenated into a single block ($\Pi_0$), forming an obligation $\Tau_{\mathrm{req}}$, without the possibility for a single party to include additional transactions afterwards.
As in FOCIL, a transaction is included in $\Tau_{\mathrm{req}}$ if at least one proposer includes it in its partial block, again relying on a one-of-$\kappa$-honest assumption.
In contrast to FIL and FOCIL, proposers receive direct rewards for transactions included in their partial blocks that are ultimately incorporated into the aggregated block.
Each proposer retains at most the per-transaction tip~$\ftiptx{\tau}$ for transactions in its own partial block, not considering MEV extraction.
Consequently, omitting a transaction~$\tau$ induces an opportunity cost equal to the forgone tip.
As a result, the required per-member bribe increases by approximately the transaction tip relative to FOCIL, which is around \num{0.01}~\euro{} at the median-tip transaction in late~2025, and is therefore only marginally higher than in FOCIL.

Auction-Based Inclusion Lists (AUCIL)~\cite{aucil} (Wadhwa et al.) is a $\Pi_0$ mechanism that combines committee-based construction with an explicit auction for aggregating partial inclusion lists. 
In the first phase, a committee of size $\kappa = 32$ constructs partial input lists.
Each committee member is assigned a subset of transactions via a shared allocation mechanism, such that honest inclusion of the assigned transactions forms a correlated equilibrium.
While inclusion of the assigned transactions is a rational strategy given transaction tips, compliance is not verified.
In the second phase, each committee member aggregates the received partial lists and submits a bid proportional to the number of included lists, inducing competition over the final $\Tau_{\mathrm{req}}$.
The block proposer is required to select the aggregate with the highest bid and include the corresponding transaction set in its block.
This design incentivizes inclusion of committee contributions, since omitting lists reduces bid strength and thus expected reward. 
Nevertheless, AUCIL does not provide enforceable accountability for transaction exclusion.
Appendix~\ref{appendix:proof_aucil} shows that even under this mechanism, which we consider the most effective among mechanisms based solely on transaction fees, a transaction can be censored at an expected cost of approximately \num{3.30}~\euro{} per slot under current Ethereum conditions.

\vspace{-3mm}
\subparagraph{Comparison.}
All four mechanisms share the same structural limitation: the absence of accountability in transaction selection keeps the per-slot cost of censorship low, with costs primarily determined by transaction tips and committee size.
We summarize censorship costs over one block, two blocks, and ten blocks (two minutes) for each mechanism in \Cref{cr:cost_all}. 
By enabling verifiability of transaction censorship and attribution of censorship to a specific party, misbehavior can be penalized, for instance through block invalidation and the associated loss of revenue.
With \textsc{FairFIL}, we present a mechanism for detecting transaction censorship, attributing it to a single party, and enabling penalization.
This approach leads, as indicated in \Cref{cr:cost_all}, to significantly higher costs for sustained censorship in practice.

\begin{table}[!ht]
\centering
\begin{tabular}{c l|c|ccc|l}
& \hspace{-6mm}\textbf{Mechanism} & \textbf{Committee} & \multicolumn{3}{c|}{\textbf{Censorship Costs (cumulated)}} & \textbf{Proof} \\
& & & \textbf{1 slot} & \textbf{2 slots} & \textbf{10 slots} & \\
\hline
& \hspace{-6mm}\textbf{Baseline}
& ---
& $\approx \num{0.10}~\text{\euro}$
& $\approx \num{0.20}~\text{\euro}$
& $\approx \num{1}~\text{\euro}$
& App.~\ref{appendix:proof_ethereum} \\
\hline
& \hspace{-6mm}\textbf{Prior Work}
& (constructing)
&
&
&
& \\
& \hspace{-4mm}FIL~\cite{EIP7547}
& 1
& $\approx \num{0.11}~\text{\euro}$
& $\approx \num{0.22}~\text{\euro}$
& $\approx \num{1.10}~\text{\euro}$
& App.~\ref{appendix:proof_fil} \\
& \hspace{-4mm}\textbf{FOCIL}~\cite{thiery2024eip7805}
& $16$
& $\leq \num{1.54}~\text{\euro}$
& $\leq \num{3.08}~\text{\euro}$
& $\leq \num{15.40}~\text{\euro}$
& App.~\ref{appendix:proof_focil} \\
& \hspace{-4mm}MCP~\cite{garimidi_concurrent_2025}
& $16$
& $\leq \num{1.60}~\text{\euro}$
& $\leq \num{3.20}~\text{\euro}$
& $\leq \num{16}~\text{\euro}$
& App.~\ref{appendix:proof_mcp} \\
& \hspace{-4mm}AUCIL$^*$~\cite{aucil}
& $32$
& $\leq \num{3.30}~\text{\euro}$
& $\leq \num{6.60}~\text{\euro}$
& $\leq \num{33}~\text{\euro}$
& App.~\ref{appendix:proof_aucil} \\
\hline
& \hspace{-6mm}\textbf{FairFIL}
& (verifying)
&
&
&
& \\
\multirow{2}{*}{\hspace{0mm}\rotatebox[origin=c]{90}{\underline{Attack}}\hspace{-8mm}}
\vspace{1mm}
& \hspace{2mm}Suppression
& $\approx \num{30000}$
& $\approx \num{0.10}~\text{\euro}$
& $\approx \num{50.34}~\text{\euro}$
& $\approx \num{402}~\text{\euro}$
& Sec.~\ref{sec:suppression} \\
& \hspace{2mm}\textbf{Invalidation}$^{**}$
& $\approx \num{30000}$
& $\approx \num{0.10}~\text{\euro}$
& $\approx \num{28.30}~\text{\euro}$
& $\approx \num{141}~\text{\euro}$
& Sec.~\ref{sec:prop_definitions} \\
\end{tabular}
\vspace{2mm}
\caption{Costs of censoring a median-tip transaction under the mechanisms proposed and surveyed in this paper.
All values are based on empirical measurements from Sep. to Dec. 2025 (\Cref{sec:empirical_baseline}).
Suppression and invalidation are the two ways in which censorship in \textsc{FairFIL} remains feasible at the stated cost; we describe both in detail in~\Cref{sec:prop_definitions}.\\
$^*$ For AUCIL, Wadhwa et al.~\cite{aucil} provide a tighter bound; in practice, bribes may be cheaper.\\
$^{**}$ Invalidating a block can be significantly more expensive due to MEV extraction.}
\label{cr:cost_all}
\end{table}

\vspace{-5mm}
\section{Fair Forward Inclusion Lists}
\label{sec:fairfil}
We present \textsc{FairFIL}, an accountable $\Pi_1$ censorship resistance mechanism.
Block assembler~$\beta_s$ retains full authority over the construction of block~$B_s$ and MEV extraction, but must publish every excluded transaction in a list $\textsc{FairFIL}_s$ alongside the block.
A committee of validators verifies that $\textsc{FairFIL}_s$ is complete with respect to their mempool view; if a majority finds a discrepancy, block~$B_s$ is rejected and the assembler forfeits the full block reward~$\Rmev{B_s}$.
The subsequent assembler~$\beta_{s+1}$ must include every transaction listed in $\textsc{FairFIL}_s$ in block~$B_{s+1}$; any omission forfeits block reward~$\Rmev{B_{s+1}}$.
We first describe how $\beta_s$ constructs $\textsc{FairFIL}_s$, then how validators verify its correctness and enforcement.
Informally, $\textsc{FairFIL}_s$ is correct if it contains exactly those transactions censored from~$B_s$ that remain includable in~$B_{s+1}$.
We formalize and prove these requirements via four accountability properties in \Cref{sec:prop_definitions}.

\subsection{\textsc{FairFIL} Construction}
\label{sec:construction}

\Cref{fairfil:builder} details how block assembler $\beta_s$ produces block~$B_s$ and $\textsc{FairFIL}_s$.
At a high level, $\beta_s$ performs two tasks: fulfilling the inclusion obligation from $\textsc{FairFIL}_{s-1}$, and publishing a complete disclosure of all transactions excluded from block~$B_s$.

\vspace{-1mm}
\subparagraph{Enforcing FairFIL (lines~1--3).}
$\textsc{FairFIL}_{s-1}$ is a hard constraint on block~$B_s$: every listed transaction must appear in~$B_s$.
Block assembler $\beta_s$ integrates the sequence into the block-construction strategy (line~1); lines~2--3 verify that this constraint is satisfied.
Any violation causes~$B_s$ to be rejected (see \Cref{sec:verification}), forfeiting the full block reward~$\Rmev{B_s}$.

\subparagraph{Constructing FairFIL (lines~4--24).}
The central challenge is to construct an inclusion list $\textsc{FairFIL}_s$ that is complete (covering every censored transaction), executable (each listed transaction must remain valid against the post-state of block~$B_s$), and verifiable (validators must be able to independently reconstruct the same transaction set to check for undisclosed exclusions).
Such an inclusion list is constructed in four steps.
First, block assembler $\beta_s$ considers all mempool transactions seen before the start of slot~$s$ (line~4).\footnote{Without a cutoff, $\beta_s$ would need to emulate every transaction in the mempool, including those persisting for months (e.g., spam or long-pending underpriced transactions). The specified \SI{603}{\second} window -- fifty slots (\SI{600}{\second}) plus a \SI{3}{\second} propagation margin (\Cref{sec:empirical:mempool}) -- limits this to a manageable recent window.} 
Second, these transactions are sorted according to Fair Ordering (line~5), using the per-gas tip as primary key and the transaction hash as tie-breaker.
Third, non-censored transactions --~those already contained in $B_s$ as well as those not includable in $B_s$ under Fair Ordering~-- are discarded (lines~9--17).
Fourth, to ensure that the resulting list remains executable for the subsequent block assembler, the censored transactions are executed on the post-$B_s$ state (lines~18--24).
Each transaction is applied in order, with the state updated after each successful execution.
Since block assembler~$\beta_s$ assembles block~$B_s$, the post-$B_s$ state is available at construction time (line~1).
Any transaction whose execution would violate protocol validity is discarded.

\begin{algorithm}[!ht]
\caption{Construction of block $B_s$ and $\textsc{FairFIL}_{s}$ by block assembler $\beta_s$.}
\label{fairfil:builder}
\KwIn{Previous block $B_{s-1}$, previous $\textsc{FairFIL}_{s-1}$, local mempool $\pool_s^{\beta_s}$}
\KwOut{Block $B_{s}$, $\textsc{FairFIL}_{s}$}
\tcc{Block construction with $\textsc{FairFIL}_{s-1}$ forced in (lines 2--3). \\
\textbf{Note}: Assembler $\beta_s$ is \textbf{not} required to follow the ordering of $\textsc{FairFIL}_{s-1}$.}
$B_s \gets \textit{assembleBlock}\hspace{0.5mm}(\textit{optional: } \pool_s^{\beta_s} \cup \mathrm{XOF},\textit{mandatory: } \textsc{FairFIL}_{s-1})$\;
\ForEach{$\tau \in \textsc{FairFIL}_{s-1}$}{
    \textit{assert}($\tau \in B_s$)\;
}
\tcc{Construct a transaction set according to Fair Ordering.}
$\textsc{FairOrdering} \gets \{\tau \in \pool_s^{\beta_s} \mid \mathrm{first\_seen}(\tau) \geq \mathrm{SLOT\_START} - \SI{603}{\second}\}$\;
$\textsc{FairOrdering} \gets \mathrm{sort}(\textsc{FairOrdering},\ \textit{prim: $\ftipgas{\tau}$},\ \textit{sec: tx hash},\ \textit{order: desc.})$\;
$\textsc{Censored} \gets [\,]$\;
$\mathit{state}  \gets \mathrm{state}(B_{s-1} \,||\, $\textsc{FairFIL}$_{s-1})$\,\tcp{appending $\textsc{FairFIL}_{s-1}$ on state of $B_{s-1}$.}
$\mathit{GasUsed} \gets gas(\textsc{FairFIL}_{s - 1})$\;
\ForEach{$\tau \in \textsc{FairOrdering}$}{
    \If{$\mathit{GasUsed} + 21000 > \Glimit(B_{s})$}{
        \Break\; \tcp{Minimum gas consumption is 21000~gas~\cite{wood_ethereum_2025}; no further transaction fits.}
    }
    $\mathit{state}' \gets \mathrm{execute}(\tau \text{ on } \mathit{state} )$\;
    \If{$\mathit{state}'$ is valid}{
        $\mathit{state}  \gets \mathit{state}'$\;
        $\mathit{GasUsed} \gets \mathit{GasUsed} + \mathrm{gas}(\tau)$\;
        \If{$\tau \notin B_s$}{
            append $\tau$ to $\textsc{Censored}$\;
        }
    }
}
\tcc{Determine $\textsc{FairFIL}_{s}$ by verifying includability in block $B_{s+1}$.}
$\textsc{FairFIL}_{s} \gets [\,]$\;
$\mathit{state}  \gets \mathrm{state}(B_{s})$\;
\ForEach{$\tau \in \textsc{Censored}$}{
    $\mathit{state}' \gets \mathrm{execute}(\tau \text{ on } \mathit{state} )$\;
    \If{$\mathit{state}'$ is valid}{
        $\mathit{state}  \gets \mathit{state}'$\;
        append $\tau$ to $\textsc{FairFIL}_{s}$\;
    }
}
\Return $B_s$, $\textsc{FairFIL}_{s}$\;
\end{algorithm}

\subsection{\textsc{FairFIL} Verification}
\label{sec:verification}
We use the validator set $V_s$ of slot $s$ as the committee responsible for verifying $\textsc{FairFIL}_s$.
\Cref{fairfil:verification} refines the verification procedure performed by a validator $v' \in V_s$.
A validator casts a negative vote for block $B_s$ if any of the following checks fails.
Block acceptance requires a majority of validator votes.

\vspace{-2mm}
\subparagraph{Check 1: Enforcement of $\textsc{FairFIL}_{s-1}$ (lines~1--3).}
Every transaction from $\textsc{FairFIL}_{s-1}$ must appear in block $B_s$.
Since the committee of slot $s{-}1$ verified includability (Check~2), any missing transaction is unambiguous evidence of non-compliance by block assembler~$\beta_s$.

\vspace{-2mm}
\subparagraph{Check 2: Includability of $\textsc{FairFIL}_s$ (lines~4--8).}
$\textsc{FairFIL}_s$ must comply with Fair Ordering: transactions are sorted in descending per-gas tip order (hash as tie-breaker), and executing them sequentially on the state induced by block~$B_s$ must preserve protocol validity of block~$B_{s+1}$.
The validator emulates this execution and rejects~$B_s$ if the resulting state is invalid.
A fixed ordering is essential: otherwise, $\beta_s$ could reorder preceding transactions to manipulate the state such that a censored target appears non-includable -- a claim validators could only refute by enumerating all permutations.

\begin{algorithm}[!hp]
\caption{Verification of $\textsc{FairFIL}_s$ by validator $v' \in V_s$.}
\label{fairfil:verification}

\KwIn{$B_{s-1}$, $B_s$, $\textsc{FairFIL}_{s-1}$, $\textsc{FairFIL}_s$, local mempool $\pool_s^{v'}$}
\KwOut{$\textsc{Vote}_{v'} \in \{\textsc{FairFIL}_s\text{ valid},\ B_s\text{ invalid}\}$}

\tcc{Check 1: Verifying Inclusion of all previous-FairFIL transactions in $B_s$}
\ForEach{$\tau \in \textsc{FairFIL}_{s-1}$}{
    \If{$\tau \notin B_s$}{
        \Return{$\textsc{Vote}_{v'}: B_s$ invalid}
    }
}

\tcc{Check 2: Verifying self-consistency of $\textsc{FairFIL}_s$}
\If{$\textsc{FairFIL}_s$ is not sorted correctly}{
    \Return{$\textsc{Vote}_{v'}: B_s$ invalid}
}
$\mathit{state} ' \gets \mathrm{execute}(\textsc{FairFIL}_s \text{ on } \mathrm{state}(B_s))$\;
\If{$\mathit{state} '$ is not valid}{
    \Return{$\textsc{Vote}_{v'}: B_s$ invalid}
}

\tcc{Check 3: Bounding unknown transactions; Constructing extended mempool}
$\mathrm{UnknownCount} \gets 0$; \quad $\mathrm{UnknownGas} \gets 0$\;
$\extpool_s^{v'} \gets \pool_s^{v'}$\;
\ForEach{$\tau \in \textsc{FairFIL}_s$}{
    \If{$\tau \notin \pool_s^{v'}$}{
        $\mathrm{UnknownCount} \mathrel{+}= 1$;\;
        $\mathrm{UnknownGas} \mathrel{+}= \mathrm{gas}(\tau)$;\;
        $\extpool_s^{v'} \gets \extpool_s^{v'} \cup \{\tau\}$\;
    }
}
\If{$\mathrm{UnknownCount} > \mathrm{threshold}(\pool_s^{v'}).\mathrm{count}$ \textbf{or} $\mathrm{UnknownGas} > \mathrm{threshold}(\pool_s^{v'}).\mathrm{gas}$}{
    \Return{$\textsc{Vote}_{v'}: B_s$ invalid}
}

\tcc{Check 4: Check $\textsc{FairFIL}_s$ correctness by constructing a censorship-free block}
$\textsc{FairBlock}_s \gets [\,]$\;
\begin{minipage}[t]{\linewidth}
$\textsc{LocalFairOrdering} \gets \{\tau \in \extpool_s^{v'} \mid \tau \in \textsc{FairFIL}_s \;\lor\;$\\
\hspace*{12em}$\SI{3}{\second} \leq \mathrm{SLOT\_START} - \mathrm{first\_seen}(\tau) \leq \SI{600}{\second}\}$
\end{minipage}\;
\begin{minipage}[t]{\linewidth}
$\textsc{LocalFairOrdering} \gets \mathrm{sort}(\textsc{LocalFairOrdering},\ \textit{prim: $\ftipgas{\tau}$},$\\
\hspace*{13.5em}$\textit{sec: tx hash},\ \textit{order: desc.})$
\end{minipage}\;

$\mathit{state}  \gets \mathrm{state}(B_{s-1} \,||\, $\textsc{FairFIL}$_{s-1})$\,\tcp{appending $\textsc{FairFIL}_{s-1}$ on state of $B_{s-1}$.}
$\mathit{GasUsed} \gets gas(\textsc{FairFIL}_{s - 1})$\;

\ForEach{$\tau \in \textsc{LocalFairOrdering}$}{
    \If{$\mathit{GasUsed} + 21000 > \Glimit(B_{s})$}{
        \Break\;
    }
    $\mathit{state} ' \gets \mathrm{execute}(\tau \text{ on } \mathit{state} )$\;
    \If{$\mathit{state} '$ is valid}{
        $\mathit{state}  \gets \mathit{state} '$\;
        $\mathit{GasUsed} \gets \mathit{GasUsed} + \mathrm{gas}(\tau)$\;
        append $\tau$ to $\textsc{FairBlock}_s$\;
    }
}

\ForEach{$\tau \in \textsc{FairBlock}_s$}{
    \If{$\tau \notin \textsc{FairFIL}_s \;\land\; \tau \notin B_s$}{
        \Return{$\textsc{Vote}_{v'}: B_s$ invalid}
    }
}

\Return{$\textsc{Vote}_{v'}: \textsc{FairFIL}_s$ valid}
\end{algorithm}

\vspace{-2mm}
\subparagraph{Check 3: Bounding unknown transactions (lines~9--17).}
$\textsc{FairFIL}_s$ may contain transactions unknown to a validator, either due to mempool inconsistencies or because the block assembler~$\beta_s$ deliberately injects non-propagated (XOF) transactions.
We show in \Cref{sec:empirical:mempool} that minor mempool inconsistencies can occur, particularly at home-staking validators; highly optimized builders, operating multiple nodes, are expected to maintain a complete mempool view.
Injected XOF poses two problems.
First, any XOF transaction in $\textsc{FairFIL}_s$ must be included in block~$B_{s+1}$, allowing block assembler~$\beta_s$ to extract MEV across two slots and depriving~$\beta_{s+1}$ of both block space and the opportunity to extract this MEV itself -- a dynamic~$\beta_s$ can exploit strategically.
Second, a central objective of accountability is that actions within \textsc{FairFIL} remain verifiable for users; injected XOF undermines this, as such transactions appear in $\textsc{FairFIL}_s$ without prior public visibility and can displace legitimate mempool transactions, as demonstrated in~\Cref{tab:censorship-example}.
Since validators cannot distinguish injected XOF from mempool inconsistencies, we define a private, per-validator threshold on the count and cumulative gas of unknown transactions in $\textsc{FairFIL}_s$; exceeding the threshold causes~$B_s$ to be rejected.
We specify the threshold in \Cref{sec:empirical:mempool} and argue that the size required to absorb mempool inconsistencies is too small for~$\beta_s$ to consume substantial block space or extract meaningful MEV.
Unknown transactions within the threshold bounds are added to an extended mempool to preserve dependency chains between known and unknown transactions (line~15).

\vspace{-2mm}
\subparagraph{Check 4: Accountability (lines~18--33).}
To verify that $\beta_s$ accounted for every exclusion, the validator applies Fair Ordering to the extended mempool and constructs a local `censorship-free' reference block~$\textsc{FairBlock}_s$.
Every transaction in $\textsc{FairBlock}_s$ must appear either in block $B_s$ (included) or in $\textsc{FairFIL}_s$ (censored by~$\beta_s$).
Any transaction absent from both reveals an undisclosed exclusion and causes the validator to vote against~$B_s$.

\section{Analysis of \textsc{FairFIL}}
\label{sec:properties}
In this section, we first demonstrate deployability by evaluating \textsc{FairFIL} on Ethereum mainnet (\Cref{sec:empirical}) and validate the key assumption of consistent mempools empirically (\Cref{sec:empirical:mempool}).
Then, we prove that non-compliance with \textsc{FairFIL} causes block invalidation, thereby establishing $(\mathcal{M},1)$-censorship resistance, where $\mathcal{M}$ is approximately equal to the block reward (\Cref{sec:prop_definitions}).
\Cref{censorship:cost} quantifies censorship costs across all mechanisms.

\subsection{Empirical Evaluation}
\label{sec:empirical}
We implement \textsc{FairFIL} and emulate its operation on \num{214600} Ethereum mainnet blocks from September~2025, measuring the sizes of the resulting per-slot inclusion lists and the number of censored transactions forced into the subsequent block.
To identify censored transactions, we reconstruct a censorship-free reference block for each slot and compare it to the actual Ethereum mainnet block~$B_s$, then construct $\textsc{FairFIL}_s$ according to \Cref{fairfil:builder}.

\subparagraph{Method.}
We use transactions from Flashbots' Mempool Dumpster~\cite{flashbots_mempool_dumpster_2025}, a public archive of publicly propagated transactions collected by multiple geographically distributed nodes.
For each block $B_s$, we extract every transaction first observed between \SI{3}{\second} and \SI{600}{\second} before the start of slot $s$ and sequence them according to Fair Ordering.
We execute this sequence using a modified version~\cite{fairfil_anvil} of Foundry's Anvil~\cite{foundry_anvil}, adapted to replay a specified transaction sequence against a given historical state.
We skip transactions that are not validly includable, as they are not censored in block $B_s$.
The historical state corresponds to the post-state of block $B_{s-1}$ and was obtained from a self-operated Ethereum execution-layer node~\cite{paradigmxyz_reth}.

\vspace{-1mm}
\subparagraph{Results.}
\label{fairfil:results}
While most blocks are censorship-free, we observe an average of \num{8.3} censored transactions per block across all observed blocks, amounting to \num{1.78}~million transactions censored by at least one slot during September~2025.
Of the \num{214600} blocks examined, \num{81411} blocks (\SI{38}{\percent}) yielded a non-empty $\textsc{FairFIL}_s$; averaged over all blocks, a $\textsc{FairFIL}_s$ contains \num{7.3} transactions.
\Cref{fig:censored_size} shows the sizes of $\textsc{FairFIL}_s$, grouped by block assembler.
Each non-empty $\textsc{FairFIL}_s$ consumed a median of \num{63209}~gas in the subsequent block (95th percentile (P95): \num{6057234}~gas), corresponding to \SI{0.1}{\percent} of the block's gas budget (P95: \SI{13.3}{\percent}).
Inclusion of the transactions in $\textsc{FairFIL}_s$ succeeded in every case, as we formally prove in \Cref{lemma:prop_includability}.
We validated robustness by repeating the analysis with the mempool of our own Ethereum node and obtained nearly identical results.
Of all transactions censored in slot~$s$, \SI{87.4}{\percent} were still includable in block $B_{s+1}$ and enforced by $\textsc{FairFIL}_s$; the remaining \SI{12.6}{\percent} were no longer includable.
For \SI{49.1}{\percent} of these non-includable transactions, the provided base fee fell below block $B_{s+1}$'s base fee.
Since base fees rise by at most \SI{12.5}{\percent} between two consecutive blocks~\cite{EIP-1559}, this non-includability can be prevented by accounting for this worst‑case increase upfront in the user's base‑fee bid.
A further \SI{50.6}{\percent} had been replaced by the user, and the remaining \SI{0.3}{\percent} were prevented by other user-initiated state changes.
Constructing the $\textsc{FairFIL}_s$ takes at most \SI{0.388}{\second} per slot
in our unoptimized, single-threaded implementation on an AMD EPYC~4564P.
This is well below the \SI{4}{\second} attestation deadline, indicating that construction cost is not a limiting factor for the timeliness requirements imposed by Ethereum's current slot anatomy.

\begin{figure*}[!ht]
    \centering
    \includegraphics[clip, trim=0 2cm 0 0cm, width=\textwidth]{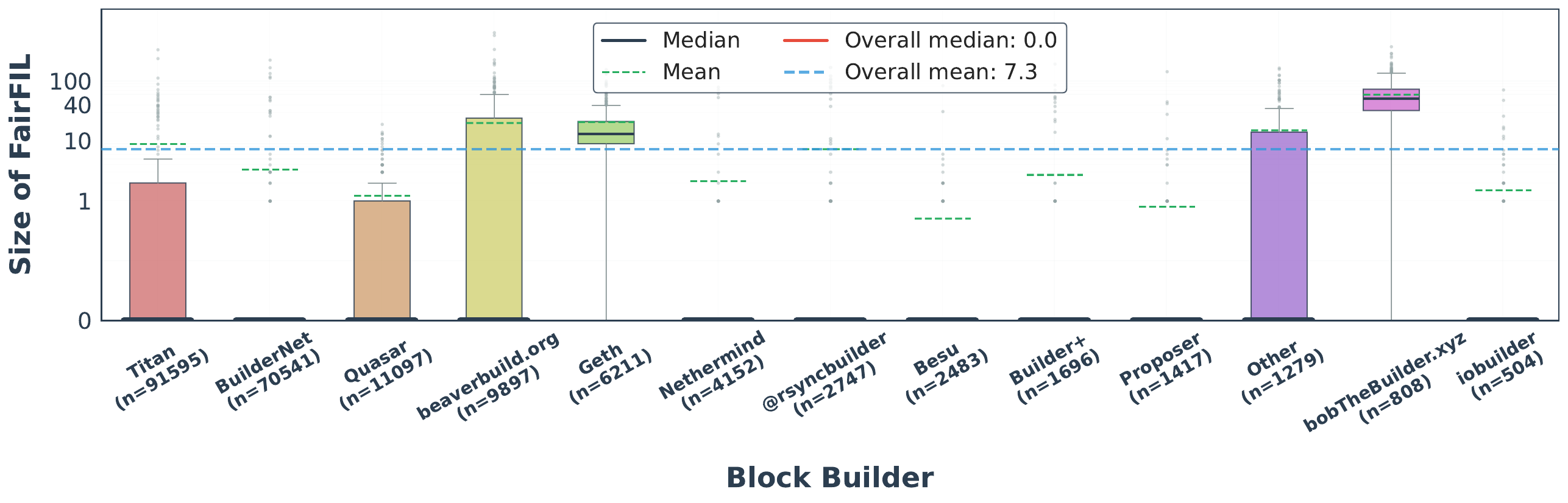}
    \caption{Size of \textsc{FairFIL} (in number of transactions) per block between blocks \num{23264600} and \num{23479200} (September~2025), grouped by block assembler.
    Note the logarithmic $y$-axis and the number $n$ of blocks proposed by each assembler during the observation period.
    ``Proposer'' reflects block assemblers that disclose no identifying information. Block assemblers with fewer than 500 proposed blocks during the observation period are aggregated into ``Other''.
    }
    \label{fig:censored_size}
\end{figure*}

\subsection{Mempool Consistency}
\label{sec:empirical:mempool}
Both the operation and the analysis of \textsc{FairFIL}'s censorship-resistance properties assume that all participants share a sufficiently consistent mempool view.
We now revisit this assumption using the September~2025 mempool dataset of the Decentralized Systems Lab at Yale University~\cite{mempool_guru}, a data source distinct from the one used in \Cref{sec:empirical} to guard against source-specific artifacts.
In particular, we need to provide evidence that block assemblers can observe every transaction the majority of validators deems relevant under Fair Ordering, and that this majority observes those transactions within the required time window.
We show that both hold sufficiently well for \textsc{FairFIL}: block assemblers can be expected to observe all transactions propagated through Ethereum's P2P network in a timely manner, and in every block, a vast majority of validators would have voted correctly on \textsc{FairFIL} compliance.

\vspace{-1mm}
\subparagraph{Measurement setup.}
The dataset provided in~\cite{mempool_guru} contains mempool observations from three nodes geographically distributed across the United States and Germany.
We combine this dataset with mempool data from our own Ethereum node, located in Karlsruhe (Germany).
Each node logs the first-observation timestamp of a transaction received through Ethereum's P2P layer.
For every transaction contained in the censorship-free blocks (cf. \Cref{sec:empirical}), we take the earliest first-observation timestamp across all four nodes as the transaction's propagation origin and then measure how many of the remaining nodes observed the transaction within the subsequent \SI{3}{\second}.

In this analysis, we do not consider transactions that were not includable in the next possible block, since spam transactions (e.g., severely underpriced ones) would otherwise distort the measurement.
The window of three seconds is motivated by prior measurements showing that block assemblers typically begin block construction about \SI{3}{\second} before slot start.
Overall, the measurement covers $n = \num{19609591}$ transactions from September~2025.

\subparagraph{Results.}
In \SI{96.12}{\percent} of all blocks, all four monitors observed every transaction that should appear in a censorship-free block.
Moreover, for every block, at least two of the four nodes observed all such transactions.
When deviations occurred, they were typically minimal: for each individual node, about \SI{2}{\percent} of blocks were missing at most a single required transaction.
For the remaining \SI{1.88}{\percent} of blocks, closer inspection suggests that the deviations are primarily caused by temporary disruptions at individual nodes, since missed transactions typically occur in bursts across multiple consecutive blocks.
Propagation latency across the four nodes, defined as the time between the first observation at any node and the first observation at the last receiving node, was below \SI{0.468}{\second} for \SI{95}{\percent} of transactions.
\Cref{fig:mempool_consistency} in Appendix~\ref{appendix:mempool_consistency} shows the resulting per-node consistency distribution.
Two quantitative observations follow.

\begin{observation}[Single-node coverage]
\label{obs:consistency}
A single reachable node observes at least \SI{90}{\percent} of the public transactions that Fair Ordering requires for inclusion in block~$B_s$, within \SI{3}{\second} of first observation by any of the monitoring nodes.
\end{observation}

\begin{observation}[Multi-node coverage]
\label{obs:consistency_merged}
A block assembler that merges the mempool views of three independent nodes observes every public transaction that Fair Ordering requires for inclusion in block~$B_s$; using two nodes already covers at least \SI{99.79}{\percent} of these transactions.
\end{observation}

\vspace{-4mm}
\subparagraph{Interpretation.}\label{dis:xof_threshold}
The data suggest that this level of consistency is sufficient for \textsc{FairFIL}.
These results do not imply that mempool consensus is trivial. 
Instead, they provide evidence that the transaction set relevant to \textsc{FairFIL} per slot is likely observable by both a well-provisioned builder and the validator majority within the required time window, without requiring agreement on temporal ordering.
Under enshrined PBS~\cite{EIP7732}, expected in mid-2026, proposers obtain protocol-level access to a specialized builder market in which builders operate as highly optimized infrastructure rather than as individual home servers.
Under these conditions, a rational builder can be expected to operate multiple geographically distributed nodes and, by \Cref{obs:consistency_merged}, not to miss any transaction that has been known to a majority of nodes for at least \SI{3}{\second}.
Regarding enforcement, the same does not hold for validators running a single node, so \textsc{FairFIL} includes an adaptive, validator-local XOF threshold (cf.\ \Cref{fairfil:verification}, Check~3) to mitigate occasional false votes due to incomplete mempool views.
The threshold should be unknown to the block assembler and could be designed to self-adjust: rising on repeated disagreement with the majority, falling otherwise.
Empirically, a majority of validators observes all \textsc{FairFIL}-relevant transactions in all cases (\Cref{obs:consistency_merged}).
Our evaluation of the adaptive threshold shows that it remains equal to zero for most validators and close to zero for the remainder across all blocks.

\subsection{Analyzing \textsc{FairFIL}'s Censorship Resistance}
\label{sec:prop_definitions}
We introduce four properties that jointly characterize an accountable mechanism and imply censorship resistance against all protocol-deviating attacks.
We prove that \textsc{FairFIL} satisfies each property under the assumption of a bribing budget~$\mathcal{M} \leq \min ( \mathcal{R}(B_s), \mathcal{R}(B_{s+1}))$:
any violation of an accountability property invalidates a block either at construction (rejecting block~$B_s$) or at enforcement time (rejecting block~$B_{s+1}$).
Property~\ref{prop:includability} ensures that the obligation to include the set of required transactions~$\Tau_{\mathrm{req}}$ imposed by \textsc{FairFIL} can be fulfilled, so failing to include~$\Tau_{\mathrm{req}}$ is unambiguous non-compliance.
Property~\ref{prop:vaware} requires every censored transaction to appear in $\Tau_{\mathrm{req}}$; Property~\ref{prop:freedom} restricts $\Tau_{\mathrm{req}}$ to publicly observed transactions.
Property~\ref{prop:compliance} requires that the obligated assembler has no rational incentive to deviate from the obligation.
Enforcement rests on majority voting of $V_s$ on Checks~1--4 (\Cref{fairfil:verification}); reversing this verdict requires bribing a majority of~$V_s$ at a cost exceeding the full block reward $\Rmev{B_s}$.
Throughout, we assume the majority of validators votes correctly.
\begin{property}[Includability]
\label{prop:includability}
For each transaction $\tau \in \Tau_{\mathrm{req}}$, the block assembler in slot $s+n$ can include $\tau$ in block~$B_{s+n}$ at the position prescribed by $\Pi_n$, while preserving block validity.
\end{property}

A foundational requirement of accountability is the distinction between deliberate non-compliance and circumstances that prevent compliance.
Property~\ref{prop:includability} establishes this distinction by admitting only obligations that the corresponding block assembler can fulfill.
Without this property, the assembler could populate $\Tau_{\mathrm{req}}$ with transactions that consume list capacity but cannot be included in the enforcement block, thereby preventing inclusion of the target transaction.
Prior work~\cite{aucil} refers to this property as ``unconditional''.

\begin{lemma}
\label{lemma:prop_includability}
\textsc{FairFIL} satisfies Property~\ref{prop:includability}: every transaction in $\textsc{FairFIL}_s$ is includable in block $B_{s+1}$ at the position prescribed by $\textsc{FairFIL}_s$ without violating the validity of $B_{s+1}$.
Any transaction failing this condition would cause block~$B_s$ to be rejected.
\end{lemma}

\begin{proof}
By Check~2 of \Cref{fairfil:verification}, the includability of $\textsc{FairFIL}_s$ requires that executing the entire transaction sequence of $\textsc{FairFIL}_s$ on top of $B_s$'s induced state yields a valid state.
Since every operation in Ethereum is deterministic~\cite{buterin_ethereum_2014}, all validators compute the same state and apply the same $\textsc{FairFIL}_s$, so the execution outcome is identical across validators.
Because transactions in a sequential execution do not retroactively invalidate earlier transactions -- each transaction's validity depends only on the state induced by its predecessors -- every prefix of $\textsc{FairFIL}_s$ executes validly on top of $B_s$'s induced state.
In particular, the prefix ending in~$\tau$ executes without violating validity.
A majority of validators in~$V_s$ has verified this condition and accepted the block; otherwise $B_s$ is rejected.
Therefore, $\beta_{s+1}$ can include this prefix at least at the beginning of $B_{s+1}$ without violating the validity of the protocol.
\end{proof}
 
\begin{property}[Accountable listing]
\label{prop:vaware}
Let $\tau$ be a transaction observed by a majority of committee members by a specific point in time.
Following Fair Ordering: if $\tau$ is includable in block~$B_{s+n}$, then $\tau \in \Tau_{\mathrm{req}}$, unless including~$\tau$ would displace a higher-tipped transaction $\tau'$.
\end{property}
 
An enforceable obligation does not establish accountability if the content of the obligation cannot be verified.
Property~\ref{prop:vaware} extends accountability to the construction of $\Tau_{\mathrm{req}}$: from a consistent mempool view, any party can determine which transactions Fair Ordering mandates, thereby making any omission attributable to the corresponding assembler.

\begin{lemma}
\label{lemma:prop_listing}
\textsc{FairFIL} satisfies Property~\ref{prop:vaware}: every transaction $\tau$ that (i) is observed at least \SI{3}{\second} before the start of slot $s$ by a majority of $V_s$, (ii) is includable in both blocks $B_s$ and $B_{s+1}$, and (iii) would not displace any transaction $\tau'$ with $\ftipgas{\tau'} > \ftipgas{\tau}$, is included in either block $B_s$ or $\textsc{FairFIL}_s$.
Any omission causes block~$B_s$ to be rejected.
\end{lemma}
 
\begin{proof}
Assume, for contradiction, that $\tau$ satisfies conditions~(i)--(iii) but $\tau \notin B_s \cup \textsc{FairFIL}_s$.
Let $W \subseteq V_s$ denote the set of validators that observe $\tau$ at least \SI{3}{\second} before the start of slot~$s$; by~(i), $|W| > \frac{|V_s|}{2}$.
Consider any validator $v \in W$ executing Check~4 of \Cref{fairfil:verification}.
Since $\tau \in \pool_s^{v} \subseteq \extpool_s^{v}$ (Check~3), $\tau$ is contained in the local Fair Ordering reconstructed by~$v$.
By~(iii), no higher-tipped transaction displaces $\tau$, so the iteration over $\textsc{LocalFairOrdering}$ reaches~$\tau$; by~(ii), executing~$\tau$ yields a valid state, so~$\tau$ is admitted into $\textsc{FairBlock}_s$.
Since $\tau \in \textsc{FairBlock}_s$ but $\tau \notin B_s \cup \textsc{FairFIL}_s$, validator $v$ detects undisclosed transaction censorship and votes against $B_s$.
This argument applies to every validator in~$W$, so a majority of $V_s$, since $|W| > \frac{|V_s|}{2}$, rejects $B_s$, and block assembler~$\beta_s$ forfeits the entire block reward~$\Rmev{B_s}$.
Compensating this loss would require a bribe exceeding $\Rmev{B_s} \geq \mathcal{M}$, contradicting the budget restriction and thus the assumption $\tau \notin B_s \cup \textsc{FairFIL}_s$.
\end{proof}
 
\begin{property}[Majority observability]
\label{prop:freedom}
The obligation $T_{\mathrm{req}}$ must contain only transactions observed by a majority of committee members.
\end{property}

While Property~2 prevents omission of eligible transactions, a rational assembler may also exploit the obligation by injecting own transactions that were not publicly broadcast beforehand. This undermines accountability, since the construction of $T_{\mathrm{req}}$ is no longer publicly verifiable, and may displace censored transactions. \Cref{tab:censorship-example} illustrates the consequence: placing $\tau_{\mathrm{private}}$ in $T_{\mathrm{req}}$ leaves insufficient capacity for the censored $\tau_4$. 
Property~3 rules out such injections by restricting $T_{\mathrm{req}}$ to transactions observed by a majority of $V_s$.

\begin{lemma}
Under Ethereum's empirical mempool consistency (\Cref{sec:empirical:mempool}), FairFIL satisfies Property~3: every transaction $\tau$ contained in $\textsc{FairFIL}_s$ is observed by a majority of $V_s$; otherwise, a majority of $V_s$ rejects block $B_s$.
\end{lemma}

\begin{proof}
By \Cref{sec:empirical:mempool}, a majority of validators set the XOF threshold in Check~3 of \Cref{fairfil:verification} to zero; we assume this throughout.
Suppose, for contradiction, that some transaction $\tau \in \textsc{FairFIL}_s$ is not observed by a majority of $V_s$.
If all validators share the same mempool view, then $\tau \notin \pool_s^v$ for every validator $v \in V_s$. Check~3 therefore marks $\tau$ as unknown at every validator, and with threshold zero, every $v \in V_s$ rejects block $B_s$.
If mempools are heterogeneous, then there exists a majority $W \subseteq V_s$ with $|W| > |V_s|/2$ such that $\tau \notin \pool_s^v$ for every $v \in W$. For each $v \in W$, Check~3 marks $\tau$ as unknown, and with threshold zero, $v$ rejects $B_s$.
Hence a majority of $V_s$ rejects $B_s$.
In either case, block assembler $\beta_s$ forfeits the full block reward $\Rmev{B_s} \ge \mathcal{M}$. 
Compensating this loss would require a bribe exceeding $\Rmev{B_s}$, contradicting the budget bound.
\end{proof}

\begin{property}[Rational builder compliance]
\label{prop:compliance}
The block assembler required to include transaction set $\Tau_{\mathrm{req}}$ in block $B_s$ does so if the assembler behaves rationally.
\end{property}
 
Properties~\ref{prop:includability}--\ref{prop:freedom} make the correctness of transaction set $\Tau_{\mathrm{req}}$ publicly verifiable; accountability, however, requires that this obligation is also enforced.
If non-compliance were the rational choice for the obligated block assembler, the assembler would deviate under the rationality model, and the censorship resistance mechanism would fail at enforcement.
Compliance with $\Tau_{\mathrm{req}}$ must therefore constitute the rational choice for the obligated assembler.

\begin{lemma}
\label{lemma:prop_compliance}
\textsc{FairFIL} satisfies Property~\ref{prop:compliance}: a rational (or altruistic) block assembler $\beta_{s+1}$ includes every transaction $\tau \in \textsc{FairFIL}_s$ in block $B_{s+1}$.
Any omission of a listed transaction from~$B_{s+1}$ causes block~$B_{s+1}$ to be rejected.
\end{lemma}

\begin{proof}
Assume, for contradiction, that a rational block assembler $\beta_{s+1}$ does not include some transaction $\tau \in \textsc{FairFIL}_s$ in block $B_{s+1}$.
By \Cref{lemma:prop_includability}, the omission of~$\tau$ is a deliberate choice by $\beta_{s+1}$, not a consequence of infeasibility.
At slot~$s{+}1$, every validator $v \in V_{s+1}$ runs Check~1 of~\Cref{fairfil:verification}, which verifies that every transaction in $\textsc{FairFIL}_s$ appears in block~$B_{s+1}$.
Since, by assumption, $\tau \notin B_{s+1}$, Check~1 fails for every $v \in V_{s+1}$; the block is rejected and $\beta_{s+1}$ forfeits the entire block reward~$\Rmev{B_{s+1}}$.
Compensating this loss would require a bribe exceeding $\Rmev{B_{s+1}} \geq \mathcal{M}$, contradicting the budget restriction.
Any lower bribe would contradict the rationality assumption of~$\beta_{s+1}$.
\end{proof}

The preceding lemmas establish that any protocol-deviating censorship attempt -- whether by omitting a required transaction from $\textsc{FairFIL}_s$, injecting unauthorized transactions into $\textsc{FairFIL}_s$, or failing to enforce the resulting inclusion obligation -- causes block rejection at the cost of $\Rmev{B_s}$ (construction failure) or $\Rmev{B_{s+1}}$ (enforcement failure).
\begin{lemma}
\label{lemma:fairfil_cr_protocol}
If censorship of transaction~$\tau$ for more than one slot requires a deviation from \textsc{FairFIL}, then \textsc{FairFIL} is $\bigl(\min(\Rmev{B_s},\, \Rmev{B_{s+1}}) + \cbribe,\; 1\bigr)$-censorship-resistant.
Equivalently, by deviating from \textsc{FairFIL}, two-slot censorship of transaction~$\tau$ requires invalidating at least one of the blocks $B_s$ or $B_{s+1}$.
\end{lemma}

The proof is given in Appendix~\ref{appendix:proof_protocol}. 
Let us now focus on protocol-compliant censoring, i.e., attacks that comply with \textsc{FairFIL} but try to suppress a transaction by other means.
\label{sec:suppression}
In a \emph{transaction suppression attack}, adversary~$\mathcal{A}$ acts as an ordinary user and propagates high-tip filler transactions through the network to displace $\tau$ below the gas-limit cutoff of the Fair Ordering.
We show that, for nearly all transactions, this attack is more expensive than block invalidation.

\vspace{-1mm}
\subparagraph{Intuition.}
Acting as an ordinary user, adversary $\mathcal{A}$ propagates transactions through the network such that they appear in regular mempools, where \textsc{FairFIL} treats them equivalently to any other transaction.
By propagating a filler transaction $\tau'$ ranked strictly above $\tau$ under Fair Ordering -- for simplicity, we assume a single $\tau'$ suffices and disregard the per-transaction gas cap~\cite{EIP7825} -- $\mathcal{A}$ displaces transaction $\tau$ below the block's gas limit $\Glimit(B_{s})$.
Transaction $\tau$ is consequently not considered censored and therefore not included in $\textsc{FairFIL}_s$; block assembler $\beta_{s+1}$ is thus not obligated to include $\tau$ in block~$B_{s+1}$, allowing $\tau$ to be censored at low cost in~$B_{s+1}$.
The attack is rule-conforming: \textsc{FairFIL} is verified and enforced as specified.
In the following, we derive the $(\mathcal{M}, \sigma)$-censorship resistance that \textsc{FairFIL} achieves against this attack.

\vspace{-1mm}
\subparagraph{Formalization.}
We write $\tau' \succ \tau$ if transaction $\tau'$ ranks strictly above transaction $\tau$ under Fair Ordering, that is, if $\ftipgas{\tau'} > \ftipgas{\tau}$, or if $\ftipgas{\tau'} = \ftipgas{\tau}$ and $\mathrm{hash}(\tau') < \mathrm{hash}(\tau)$.
Let $G_{\mathrm{above}}(\tau, s)$ be the cumulative gas of all mempool transactions ranked above transaction~$\tau$ at slot~$s$, and define the residual gas capacity $\Delta(\tau, s)$ above $\tau$ as the largest amount of additional gas that can be inserted above $\tau$ in Fair Ordering while $\tau$ still fits within block $B_{s}$:
\[
\Delta(\tau, s) \;=\; \max\bigl(0,\; \Glimit(B_{s}) - G_{\mathrm{above}}(\tau, s) - \mathrm{gas}(\tau)\bigr).
\]
\noindent To suppress $\tau$, adversary $\mathcal{A}$ must inject filler gas strictly exceeding $\Delta(\tau, s)$ in transactions satisfying $\tau' \succ \tau$, at a per-gas cost of at least $f^*(\tau) := \ftipgas{\tau} + \fbasegas{\tau'}$.
Matching $\ftipgas{\tau}$ rather than strictly exceeding it suffices, since $\mathcal{A}$ can outrank $\tau$ via the hash tie-breaker.
 
\begin{lemma}[Suppression cost]
\label{lemma:fairfil_cr_baseline}
\textsc{FairFIL} is at most $\bigl(f^*(\tau) \cdot \Delta(\tau, s),\; 1\bigr)$-censorship-resistant for any transaction $\tau$, propagated before slot $s$.
\end{lemma}
 
\begin{proof}[Proof Sketch]
An adversary who wants to exclude transaction $\tau$ for more than a single slot must first prevent $\tau$ from qualifying as censored in slot~$s$ -- otherwise, $\beta_{s+1}$ would be obligated to include $\tau$ via $\textsc{FairFIL}_s$.
This requires injecting a publicly broadcast filler transaction $\tau'$ whose gas usage exceeds the residual gas capacity of block~$B_s$, at a per-gas cost of at least the total per-gas transaction fee of $\tau$, so that $\tau$ no longer fits under Fair Ordering and is not classified as censored.
With no inclusion obligation in $B_{s+1}$, the adversary can then exclude $\tau$ from $B_{s+1}$ at the Ethereum baseline cost~$\Bbase_\tau$.
The full proof is given in Appendix~\ref{appendix:proof_suppression}.
\end{proof}

Unlike prior mechanisms, the per-slot cost of censorship under \textsc{FairFIL} is not uniform, but can vary across slots.
We derive this fact as follows.
At slot~$s$, no inclusion obligation applies to transaction~$\tau$, so censoring $\tau$ from block $B_s$ costs only the Ethereum baseline~$\Bbase_{\tau}$.
To extend censorship beyond a single slot, $\mathcal{A}$ must either prevent~$\tau$ from entering $\textsc{FairFIL}_s$ via the suppression attack or invalidate one of the affected blocks.
For the transaction suppression attack, $\mathcal{A}$ propagates a filler transaction~$\tau'$ in the same slot as $\tau$, w.l.o.g.\ slot~$s$, such that $\tau'$ is included by block assembler $\beta_s$ in block $B_s$.
The block $B_s$ is thereby fully utilized, leaving $\tau$ neither included nor classified as censored, and thus excluded from $\textsc{FairFIL}_s$.
Block assembler $\beta_{s+1}$ is therefore not obligated to include $\tau$ in block~$B_{s+1}$, allowing adversary $\mathcal{A}$ to achieve censorship of $B_{s+1}$ by paying $\Bbase_{\tau}$.
Under this payment, $\beta_{s+1}$ is then required to include $\tau$ in $\textsc{FairFIL}_{s+1}$, so $\mathcal{A}$ must again prevent inclusion of $\tau$ in $\textsc{FairFIL}_{s+1}$.
This coincides with the displacement required for $\textsc{FairFIL}_s$ and is necessary in every slot in which $\tau$ is to be censored for at least two blocks.
Thus, the first and last slot in which $\tau$ is censored are comparatively cheap, while all intermediate slots require repeated displacement from the preceding FairFIL. The per-slot cost is therefore

\vspace{1mm}
\(
\mathcal{C}_{\tau} \;=\; \bigl[\Bbase_{\tau},\; f^*(\tau) \cdot \Delta(\tau, s),\; f^*(\tau) \cdot \Delta(\tau, s + 1), \dots, \Bbase_{\tau} \bigr].
\)
\vspace{1mm}

\noindent We show in \Cref{censorship:cost} that, in late 2025, most Ethereum transactions already pay tips sufficient for suppression to cost nearly twice as much as block invalidation; even the mandatory base fee alone exceeds the block invalidation cost.

For block invalidation, once a block~$B_s$ is invalidated, the subsequent block assembler $\beta_{s+1}$ carries no inclusion obligation, so $\tau$ can be censored from block~$B_{s+1}$ at Ethereum baseline cost~$\Bbase_\tau$.
However, the censored transaction must then appear in $\textsc{FairFIL}_{s+1}$, requiring another block invalidation to sustain censorship in block~$B_{s+2}$. This yields an alternating cost pattern between baseline censorship costs and block invalidation.
The resulting per-slot cost pattern therefore takes one of two forms, or any combination of the underlined parts,
\[
\mathcal{C}_{\tau} \;=\; \bigl[\,
\underline{\Bbase_{\tau},\; \Rmev{B_{s+1}} \,+\, \epsilon \,+\, \cbribe},
\, \dots\,\bigr]
\quad\text{or}\quad
\mathcal{C}_{\tau} \;=\; \bigl[\,
\underline{\Rmev{B_s} + \,\epsilon\, + \,\cbribe\,,\; \Bbase_{\tau}},
\, \dots\,\bigr],
\]
depending on whether adversary $\mathcal{A}$ invalidates $B_{s+1}$ (enforcement failure) or $B_s$ (construction failure).
Together, these results determine the $(\mathcal{M},1)$-censorship resistance of \textsc{FairFIL}.

\subsection{Cost of Censorship}
\label{censorship:cost}
We now instantiate the bounds of \Cref{lemma:fairfil_cr_baseline,lemma:fairfil_cr_protocol} with empirical Ethereum data to estimate the actual costs of censorship under \textsc{FairFIL}. 
\Cref{cr:cost_all} in \Cref{sec:relatedwork} summarizes all results, including surveyed prior mechanisms with the instantiations described below.

\subparagraph{Empirical Baseline Parameters.}
\label{sec:empirical_baseline}
 
To determine the median transaction parameters, we use data from \num{194} million transactions in blocks \num{23264566}--\num{24136052} (Sep.--Dec.~2025).
\Cref{fig:eth:sota} shows costs for transaction tips, transaction fees, and block tips at a reference price of \num{2500}~\euro\ per Ether, separated into transactions propagated publicly and all transactions within the block, as well as the per-block fraction of transactions not observed in our mempool dataset.
A public transaction in this period consumes a median of \SI{40400}{gas}, pays a median per-transaction tip of $\ftiptx{\tau} \approx \num{0.01}~\text{\euro}$ to the block assembler, and incurs a total fee of approximately $\num{0.05}~\text{\euro}$ paid by the user.
A block assembler collects a median of $\Rpub{B_s} \approx \num{9.09}~\text{\euro}$ in tips per block from public transactions.
At the median, the included public transactions require approximately \SI{7.9e6}{gas}, far below the block gas limit of \SIrange{45e6}{60e6}{gas} in this period (\SI{45e6}{gas} in September 2025).
Private transactions ($\approx$~\SI{44}{\percent} of the median block's transactions) raise the assembler's revenue by a factor of~\num{3.1}, yielding a median block reward of $\Rmev{B_s} \approx \num{28.11}~\text{\euro}$.
Based on these data and corresponding to our bribing contract implementation, we set $\cbribe = \num{0.08}~\text{\euro}$ ($\approx$ \num{85000} gas at base fees)\footnote{We do not implement a fully trustless bribing setup; instead, we rely on the correct operation of an oracle. The reported \num{85000} gas correspond to the average per-participant cost of bribing 16 participants, with significant variation depending on the number of participants. The average cost per bribe increases for fewer than 16 participants and decreases for larger sets of participants.
} and $\epsilon = \num{0.01}~\text{\euro}$.
The chosen deviation incentive $\epsilon$ is small in absolute terms and, in our view, only meaningful as an incentive for actors that accumulate bribes across many slots, such as operators of many validators~\cite{ethereum-saas}.
We nevertheless adopt this conservative value to ensure that the protocol resists any short-term rational deviations.

\begin{figure*}[tb]
  \centering
  \begin{subfigure}[b]{0.24\linewidth}
    \centering
    \includegraphics[width=\textwidth]{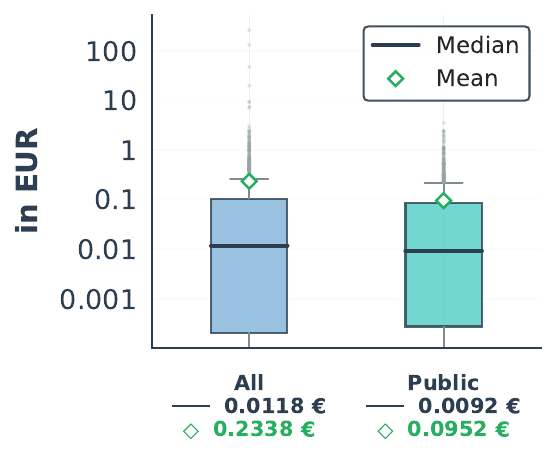}
    \captionsetup{justification=centering}
    \caption{Transaction tip}
  \end{subfigure}
  \hfill
  \begin{subfigure}[b]{0.24\linewidth}
    \centering
    \includegraphics[width=\textwidth]{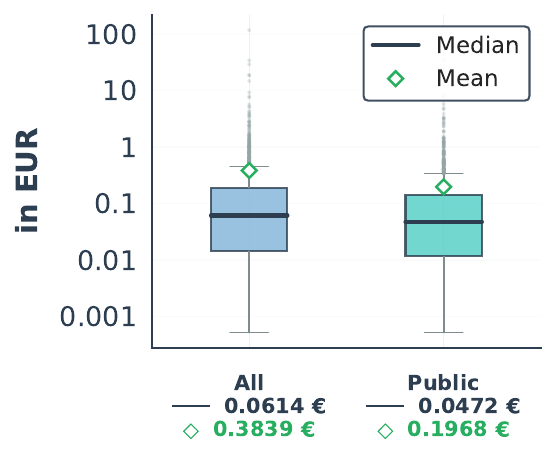}
    \captionsetup{justification=centering}
    \caption{Transaction fee}
  \end{subfigure}
  \hfill
  \begin{subfigure}[b]{0.24\linewidth}
    \centering
    \includegraphics[width=\textwidth]{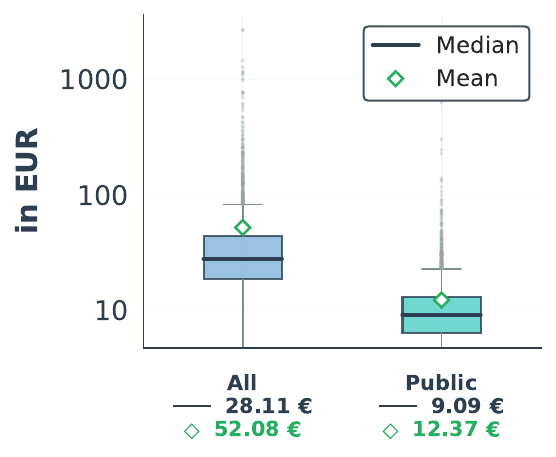}
    \captionsetup{justification=centering}
    \caption{Block tip}
  \end{subfigure}
  \hfill
  \begin{subfigure}[b]{0.24\linewidth}
    \centering
    \includegraphics[width=\textwidth]{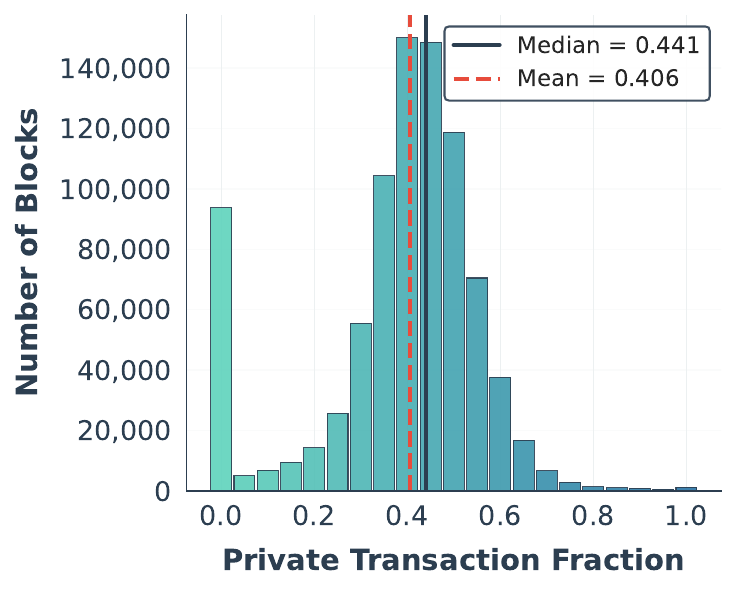}
    \vspace{-2.5mm}
    \captionsetup{justification=centering}
    \caption{\% of XOF per block}
  \end{subfigure}
  \caption{Historical transaction and block data ($n = \num{194051211}$ transactions in blocks \num{23264566}--\num{24136052}, Sep.--Dec.~2025) from the Ethereum mainnet. (a)~the transaction tip ($\ftiptx{\tau}$) paid to the assembler (in~\euro); (b)~the total transaction fee ($\ftiptx{\tau}$ + $\fbasetx{\tau}$) paid by the user (in~\euro); (c)~the block reward from tips (in~\euro); (d)~the per-block fraction of transactions that had not been observed in our mempool dataset~\cite{mempool_guru}. We use a constant exchange rate of \num{2500}~\euro\ per~Ether.}
  \label{fig:eth:sota}
\end{figure*}

\subparagraph{Quantifying Censorship Costs.}
To censor a transaction $\tau$, adversary $\mathcal{A}$ must either invalidate a block (\Cref{lemma:fairfil_cr_protocol}) or suppress the transaction (\Cref{lemma:fairfil_cr_baseline}).
Empirically, invalidating a block results in a loss of $\Rmev{B_s} \approx \num{28.11}~\text{\euro}$; accepting a bribe is therefore the rational decision for a payment of $\approx \num{28.20}~\text{\euro}$.
In practice, this value is likely higher, since a block assembler typically extracts profitable MEV beyond the tips.
To determine the cost of displacing $\tau$ from the residual gas capacity in Fair Ordering, we consider $\tau$ at the lowest rank under Fair Ordering and quantify the cost of suppressing such a transaction.
Empirically, the adversary needs a median of $\approx \num{37100000}$~gas across all observed blocks to suppress~$\tau$.
These are tip-less transactions; at $f^*(\tau) \approx \frac{\num{0.038}~\text{\euro}}{\num{40400}\,\mathrm{gas}}$, filling this capacity costs approximately $\num{34.9}~\text{\euro}$.
The same calculation for a transaction with a median tip yields $\num{43000000}$~gas and a cost of $\num{50.24}~\text{\euro}$.
Suppression is therefore, in most cases, more expensive than block invalidation.

The corresponding costs for $n$-slot censorship, along with a comparison to prior proposals, are reported in \Cref{cr:cost_all}.
Note that whereas prior proposals prevent censorship in the earliest possible slot, censoring a transaction under \textsc{FairFIL} remains cheap in the first slot, with substantially higher costs arising at alternating slots thereafter.
For example, to censor a transaction for two minutes (ten slots), Ethereum's upcoming censorship resistance mechanism~\cite{ethereum_hegota_2026}, FOCIL, requires ten times the per-slot bribery cost (total median: $\num{15.40}~\text{\euro}$), whereas \textsc{FairFIL} requires five block invalidations interleaved with five baseline Ethereum censorship bribes (total median: $\num{141}~\text{\euro}$).

\section{Discussion}
\label{sec:discussion}

\subsection{On Fair Ordering and $\Pi_1$ Mechanisms}
\label{discussion:pi1}
\textsc{FairFIL} uses a tip-based reference ordering over publicly observed transactions rather than a receive-time-based fairness notion such as the Aequitas family~\cite{kelkar_orderfairness_2020}.
Imposing such an ordering on transactions -- together with preventing private communication channels -- would arguably come closer to the ideal of a fair smart-contract platform.
In Ethereum, however, the effects of such constraints on MEV extraction and participant behavior remain, to the best of our knowledge, insufficiently understood.
This uncertainty matters because private transactions play a major role in practice.
In our measurement period from September to December~2025, they consumed \num{2.86}$\times$ as much gas as public transactions.
Consistent with prior work~\cite{heimbach_ethereums_2023,yang_decentralization_2025}, this indicates extensive MEV extraction and suggests that private transactions may displace public ones when blocks are fully utilized.
These observations motivate our choice of a $\Pi_1$ mechanism, although \textsc{FairFIL} could also be instantiated as a $\Pi_0$ mechanism with minor timing adjustments.
A $\Pi_0$ mechanism would constrain the current block assembler with respect to all transactions includable in the current slot.
\Cref{tab:censorship-example} illustrates the tension by example: enforcing the Fair Ordering over public transactions in full would leave insufficient capacity for the privately submitted transaction \(\tau_{\text{private}}\).
Relaxing to $\Pi_1$ does not fully eliminate this concern, as a large $\textsc{FairFIL}_s$ may still burden the next assembler.
Empirically, however, this obligation is much smaller than what a $\Pi_0$ mechanism would impose on the current block.
Even at the 95th percentile, the inclusion obligation induced by $\textsc{FairFIL}_s$ occupies only about one seventh of the subsequent block's gas capacity, whereas a $\Pi_0$ design would consume roughly half of the current block's capacity.
Given the gas-consumption skew toward private transactions, an accountable $\Pi_0$ mechanism would therefore already constrain MEV extraction -- an effect that warrants further study.
We therefore argue that \textsc{FairFIL}'s guarantee -- either time-bounded transaction inclusion or censorship costs at the level of one invalidated block -- meaningfully improves network reliability while remaining compatible with Ethereum's current incentive model and slot anatomy.

\subsection{Absence of Altruism}
\label{discussion:altruism}
Our behavior model assumes the absence of altruistic participants.
However, prior work~\cite{ethereum_altruistic_proposers_2026} has shown their existence and estimates their prevalence in early~2025 to be at most \SI{1.55}{\percent} of validators.
Relative to the current state of Ethereum, this implies that in roughly every 67th block, an adversarial bribe fails, even if accepting the bribe would be the rational decision for the block assembler.
In particular, in an uncongested network, this yields an effective censorship delay of 13.4 minutes, which we consider insufficient for time-critical applications.
For construction-based mechanisms (cf.~\Cref{sec:relatedwork}), this further implies that in committees of size $\kappa = 128$ under a one-of-$\kappa$-honest assumption, it is likely that at least one altruistic validator is present.
However, at such scales, the influence of an individual committee member becomes small, and even then inclusion of a censored transaction cannot be guaranteed.
As a result, users have no guarantee that their transaction will either be included or that exclusion requires a high-priced bribe by an adversary.
Like AUCIL, \textsc{FairFIL} does not rely on altruistic behavior, but instead assumes predominantly rational participants while maximizing the bribe required for successful censorship.

\section{Conclusion}
\label{sec:conclusion}
Accountability for exclusion decisions -- requiring every block assembler to publicly disclose all censored transactions in a verifiable and attributable manner -- can substantially raise the cost of censoring a transaction in Ethereum.
We established this claim through three contributions.
First, we formalized accountability via four properties whose joint satisfaction anchors censorship costs to the full block reward, rather than to the per-transaction tip or committee size as in prior approaches.
Empirically, this yields an order-of-magnitude increase relative to prior work such as \textsc{FOCIL}, Ethereum's upcoming censorship resistance mechanism.
Second, we presented \textsc{FairFIL} as an accountable $\Pi_1$ mechanism and proved that, under a fully rational participant model, a deviation from \textsc{FairFIL} either causes block invalidation or incurs costs of comparable magnitude.
Further, \textsc{FairFIL} requires no changes to Ethereum's existing incentive model.
Third, we implemented \textsc{FairFIL} and empirically evaluated it on \num{214600} Ethereum mainnet blocks: mempool consistency across validators appears to suffice for \textsc{FairFIL} operation, two-slot censorship costs rise from \num{3.08}~\text{\euro} (\textsc{FOCIL}) to \num{28.30}~\text{\euro} (\textsc{FairFIL}), and the median inclusion obligation imposed on the subsequent block remains small, leaving the assembler's MEV-extraction freedom largely intact.
\newpage
\bibliography{lipics-v2021-sample-article}

\appendix

\section{Notation and Abbreviations}
 
\begin{table}[!ht]
\centering
\small
\begin{tabular}{l|l}
Symbol / Abbreviation & Meaning \\
\hline
\multicolumn{2}{l}{Entities and sets} \\
$V$, $V_s$ & validator set; slot-$s$ committee \\
$v_s^{P}$ & proposer of slot $s$ ($v_s^{P} \in V_s$) \\
$\beta_s$ & block assembler of slot $s$ \\
$B_s$ & block produced in slot $s$ by $\beta_s$\\
$\mathrm{state}(B_s)$ & execution state after executing $B_s$ \\
$\pool_s^{n}$ & mempool view of node $n$ at slot $s$ \\
$\extpool_s^{v}$ & validator $v$'s extended mempool (\Cref{fairfil:verification}, Check~3) \\
\hline
\multicolumn{2}{l}{Transactions and fees} \\
$\tau$ & transaction \\
$\mathrm{gas}(\tau)$ & gas consumption of $\tau$ \\
$\ftipgas{\tau}$ & per-gas tip of $\tau$, paid to block assembler $\beta_s$ \\
$\ftiptx{\tau}$ & per-transaction tip, $\ftiptx{\tau} := \ftipgas{\tau} \cdot \mathrm{gas}(\tau)$ \\
$\fbasegas{\tau}$ & per-gas base fee of $\tau$ (burned) \\
$\fbasetx{\tau}$ & per-transaction base fee, $\fbasetx{\tau} := \fbasegas{\tau} \cdot \mathrm{gas}(\tau)$ \\
$\cbribe$ & on-chain bribing-contract overhead (cost)\\
$\epsilon$ & protocol deviation incentive (cost)\\
\hline
\multicolumn{2}{l}{Block and reward quantities} \\
$\Glimit(B_s)$ & gas limit of block $B_s$ \\
$\Rmev{B_s}$ & total block reward of $B_s$ (transaction tips + MEV) \\
$\Rpub{B_s}$ & tip revenue from public transactions in $B_s$ \\
$\Delta(\tau, s)$ & residual gas capacity above $\tau$ at slot $s$ \\
\hline
\multicolumn{2}{l}{Censorship (Resistance) notions} \\
$\Bbase_{\tau}$ & current per-slot cost of excluding $\tau$\\
$\mathcal{A}$ & adversary aiming to censor $\tau$ \\
$\mathcal{M}$ & adversary's per-slot bribing budget \\
$(\mathcal{M},\sigma)$-CR & censorship resistance notion (\Cref{def:cr}) \\
$\Pi_n$ & CR mechanism for $n$-slot censorship mitigation (\Cref{def:cr_mechanism}) \\
$\kappa$ & committee size of a $\Pi_n$ mechanism \\
$\Tau_{\mathrm{req}}$ & transactions required by $\Pi_n$ in a specified block \\
$\textsc{FairFIL}_s$ & inclusion list / \textsc{FairFIL} of slot $s$ \\
$\textsc{FairBlock}_s$ & local censorship-free block of a validator \\
\end{tabular}
\vspace{2mm}
\caption{Symbols and abbreviations used throughout this paper.}
\label{tab:notation}
\end{table}

\section{Mempool Consistency}
\label{appendix:mempool_consistency}
\Cref{fig:mempool_consistency} reports per-node mempool consistency relative to the censorship-free blocks constructed in~\Cref{sec:empirical}.
To measure mempool consistency, we proceed as follows.
For every censorship-free block, we consider the set of transactions contained therein
and record the earliest first-observation timestamp of each transaction across all monitors.
We then determine, for each monitor, the fraction of these transactions observed within \SI{3}{\second} of that timestamp.
Each data point on the vertical axis therefore corresponds to a single censorship-free block and quantifies the share of the required transaction set that was timely available at the respective monitor.
The horizontal axis distinguishes the three publicly available monitors of~\cite{mempool_guru} in Dallas~(US), Frankfurt am Main~(GER), and Oregon~(US).
In the median, every monitor observes all transactions required to reconstruct a censorship-free block.
Deviations occur but are confined to a small share of blocks and remain quantitatively small; we characterize their magnitude via the 1st and 2nd percentile of the per-node distribution.
For instance, for Frankfurt am Main, P2~=~\SI{98.55}{\percent} indicates that in \SI{98}{\percent} of the observed censorship-free blocks, no more than \SI{1.45}{\percent} of required transactions were missing within the three-second time window.
The lower P1 of \SI{89.36}{\percent} is likely caused by a short-lived monitoring outage, as the missed transactions cluster across several consecutive blocks.
Dallas and Oregon exhibit the same pattern at slightly higher percentile values.

\vspace{15mm}
 \begin{figure}[!ht]
    \centering
    \includegraphics[width=0.64\linewidth]{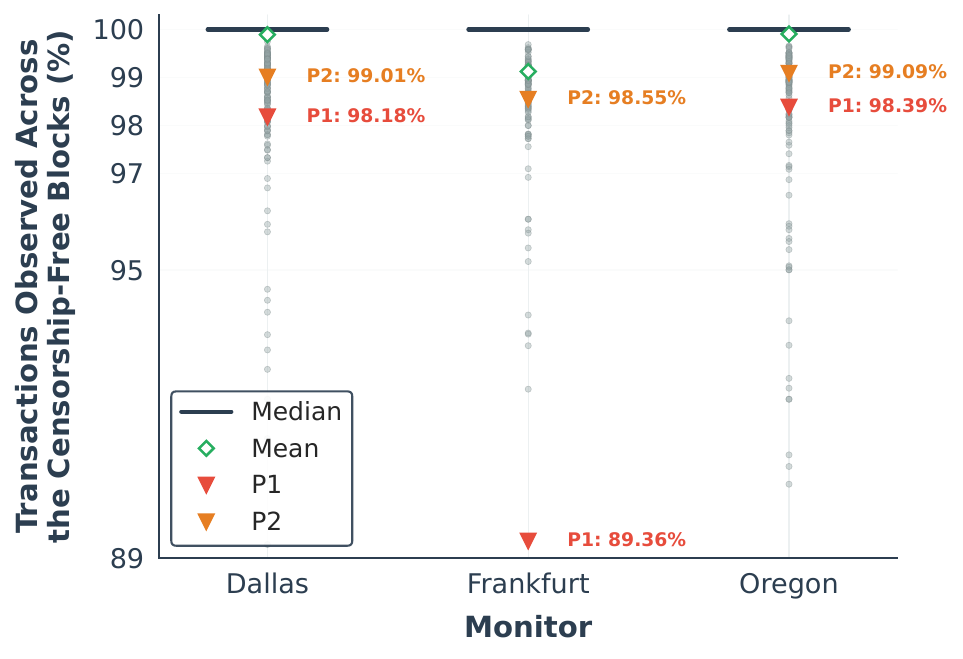}
    \caption{Per-node mempool consistency during September 2025 at the three monitors
of~\cite{mempool_guru}.
The vertical axis shows, per censorship-free block, the fraction of required transactions
observed within \SI{3}{\second} of their earliest first-observation across all monitors.}
    \label{fig:mempool_consistency}
\end{figure}

\section{Proof of \Cref{lemma:EthereumCR} (Ethereum baseline)}
\label{appendix:proof_ethereum}
\begin{proof}
Let $\mathcal{M} := \ftiptx{\tau} + \cbribe$.
\Cref{lemma:EthereumCR} consists of two claims, which we prove in turn.
First, we show that an adversary with per-slot budget $\mathcal{M} + \epsilon$ can exclude transaction~$\tau$ from every block, establishing the per-slot exclusion cost~$\Bbase_{\tau}$.
Second, we show that no adversary with per-slot budget~$\mathcal{M}$ can delay~$\tau$ even by a single slot, establishing $(\mathcal{M},\,0)$-censorship resistance.
 
\begin{description}
\item[Per-slot budget $\mathcal{M} + \epsilon$ suffices to censor~$\tau$ in every slot.]
Adversary $\mathcal{A}$ bribes the block assembler~$\beta_s$ through the bribing contract~\cite{sookitoth_bribers_2025} to exclude~$\tau$ from block~$B_s$.
Of the budget $\mathcal{M} + \epsilon = \ftiptx{\tau} + \cbribe + \epsilon$, the contract consumes the overhead~$\cbribe$, leaving a bribe of~$\ftiptx{\tau} + \epsilon$ paid to~$\beta_s$ upon observed exclusion of~$\tau$ from~$B_s$.
By including~$\tau$, $\beta_s$ would earn the per-transaction tip~$\ftiptx{\tau}$; by excluding~$\tau$ and accepting the bribe, $\beta_s$ earns~$\ftiptx{\tau} + \epsilon$, which strictly exceeds the inclusion payoff for any~$\epsilon > 0$.
A rational block assembler therefore excludes~$\tau$.
$\mathcal{A}$ can repeat this attack in every slot at the same per-slot cost~$\Bbase_{\tau} = \ftiptx{\tau} + \cbribe + \epsilon$, so~$\tau$ is never included.
 
\item[Budget $\mathcal{M}$ does not suffice to delay~$\tau$ even by a single slot.]
Assume, for contradiction, that an adversary with per-slot budget~$\mathcal{M}$ can cause $\tau$ to be excluded from $B_s$.
Then the bribe paid to~$\beta_s$ is at most $\mathcal{M} - \cbribe = \ftiptx{\tau}$.
Accepting this bribe yields~$\beta_s$ exactly the same payoff as including~$\tau$ and earning the tip~$\ftiptx{\tau}$.
By \Cref{sec:cr_framework}, in case of indifference, a rational~$\beta_s$ follows the protocol and includes~$\tau$ in~$B_s$, contradicting the assumed exclusion.
Hence no adversary with per-slot budget~$\mathcal{M}$ can delay~$\tau$ for any slot.
Therefore, Ethereum in its current state is $(\mathcal{M},\,0)$-censorship-resistant.\qedhere
\end{description}
\end{proof}
  
\section{Proof of \Cref{lemma:fairfil_cr_protocol} (Protocol-deviation cost)}
\label{appendix:proof_protocol}
\begin{proof}
Let $\mathcal{M} := \min\bigl(\Rmev{B_s},\, \Rmev{B_{s+1}}\bigr) + \cbribe$ and, without loss of generality, assume $\Rmev{B_{s+1}} \leq \Rmev{B_s}$ so that $\mathcal{M} = \Rmev{B_{s+1}} + \cbribe$.
We further assume the condition of \Cref{lemma:fairfil_cr_protocol}: two-slot censorship of~$\tau$ requires a deviation from \textsc{FairFIL} (i.e., suppression is infeasible within the considered budget).
We further assume for all proofs that $\Bbase_{\tau} \leq \mathcal{M}$.
 
\subparagraph*{Enumeration of deviation strategies.}
For two-slot censorship, transaction~$\tau$ must be censored from both~$B_s$ and~$B_{s+1}$.
Under the assumption, this requires one of the following protocol deviations:
\begin{description}
\item[Strategy (a): Construction failure.] Adversary~$\mathcal{A}$ bribes block assembler~$\beta_s$ to construct a non-compliant $\textsc{FairFIL}_s$ that omits~$\tau$ (Property~\ref{prop:vaware}). A majority of~$V_s$ rejects~$B_s$ (\Cref{lemma:prop_listing}); $\beta_s$ forfeits the full block reward~$\Rmev{B_s}$.
\item[Strategy (b): Enforcement failure.] Adversary~$\mathcal{A}$ bribes~$\beta_{s+1}$ to omit~$\tau \in \textsc{FairFIL}_s$ from~$B_{s+1}$ (Property~\ref{prop:compliance}). A majority of~$V_{s+1}$ rejects~$B_{s+1}$ (\Cref{lemma:prop_compliance}); $\beta_{s+1}$ forfeits~$\Rmev{B_{s+1}}$.
\end{description}
 
\begin{description}
\item[Budget $\mathcal{M} + \epsilon$ suffices to censor~$\tau$ for two slots.]
With $\Rmev{B_{s+1}} \leq \Rmev{B_s}$, adversary~$\mathcal{A}$ uses Strategy~(b):
\begin{enumerate}
\item In slot~$s$, $\mathcal{A}$ bribes~$\beta_s$ at the Ethereum baseline cost~$\Bbase_{\tau}$ (\Cref{lemma:EthereumCR}) to exclude~$\tau$ from~$B_s$. As~$\tau$ remains includable in~$B_{s+1}$ under Fair Ordering, $\tau$ enters $\textsc{FairFIL}_s$ by \Cref{lemma:prop_listing}.
The slot-$s$ expenditure is~$\Bbase_{\tau} \leq \mathcal{M}$.
\item In slot~$s+1$, $\mathcal{A}$ bribes~$\beta_{s+1}$ through the bribing contract to omit~$\tau \in \textsc{FairFIL}_s$ from~$B_{s+1}$, resulting in a loss of block reward $\Rmev{B_{s+1}}$. 
The bribe must compensate the forfeiture of~$\Rmev{B_{s+1}}$ and strictly exceed it by some $\epsilon > 0$ to make omission strictly preferable for a rational~$\beta_{s+1}$ (\Cref{sec:cr_framework}); the contract additionally consumes the overhead~$\cbribe$. 
The slot-$s{+}1$ expenditure is therefore~$\Rmev{B_{s+1}} + \cbribe + \epsilon = \mathcal{M} + \epsilon$.
\end{enumerate}

\vspace{1mm}
The per-slot peak expenditure for two-slot censorship is~$\mathcal{M} + \epsilon$.
By Property~\ref{prop:compliance}, $B_{s+1}$ is rejected, and no obligations apply to the subsequent block assembler $\beta_{s+2}$.
The resulting delay is $\sigma \geq 2$.
Thus, \textsc{FairFIL} is at most $(\mathcal{M},\,1)$-censorship-resistant.
The case $\Rmev{B_s} \leq \Rmev{B_{s+1}}$ is argumentatively identical, requiring an attack via Strategy~(a) and the payment of the baseline cost in slot~$s+1$.

\vspace{1mm}
\item[Budget $\mathcal{M}$ does not suffice to censor~$\tau$ for two slots.]
Assume adversary~$\mathcal{A}$ has per-slot budget~$\mathcal{M}$ and attempts Strategy~(b).
Of the budget, the bribing contract consumes the overhead~$\cbribe$, leaving at most~$\Rmev{B_{s+1}}$ payable to~$\beta_{s+1}$.
By accepting, $\beta_{s+1}$ forfeits~$\Rmev{B_{s+1}}$, yielding the same net payoff as honest compliance.
By the rationality model of \Cref{sec:cr_framework}, in case of indifference~$\beta_{s+1}$ follows the protocol and includes~$\tau$ in~$B_{s+1}$.
For Strategy~(a), the analogous bribe to~$\beta_s$ falls strictly short of the higher loss~$\Rmev{B_s} \geq \Rmev{B_{s+1}}$, so~$\beta_s$ strictly prefers compliance.
Hence no deviation occurs within budget~$\mathcal{M}$, and by the assumption, two-slot censorship is infeasible.
Therefore, \textsc{FairFIL} is $(\mathcal{M},\,1)$-censorship-resistant.\qedhere
\end{description}
\end{proof}

\section{Proof of \Cref{lemma:fairfil_cr_baseline} (Suppression cost)}
\label{appendix:proof_suppression}
\begin{proof}
Let $\mathcal{M} := f^*(\tau) \cdot \Delta(\tau, s)$.
We assume that under Ethereum's transaction fee mechanism~\cite{EIP-1559}, blocks are typically not fully utilized -- by design, the average target fullness is approximately half of the gas limit~\cite{EIP-1559}.
Consequently, an adversary populating residual gas capacity above~$\tau$ in Fair Ordering must contribute gas that would otherwise remain unused, so $\Delta(\tau, s)$ is large relative to the gas of a single transaction.
Based on our observations in \Cref{sec:empirical}, we therefore assume throughout that $\Bbase_{\tau} < f^*(\tau) \cdot \Delta(\tau, s) = \mathcal{M}$.
 
\begin{description}
\item[Budget $\mathcal{M} + \epsilon$ suffices to censor~$\tau$ for two slots.]
Adversary $\mathcal{A}$ proceeds as follows:
$\mathcal{A}$ propagates a filler transaction~$\tau'$ through the P2P layer such that $\tau' \succ \tau$ under Fair Ordering and the gas usage of~$\tau'$ strictly exceeds~$\Delta(\tau, s)$.
The parameter $\epsilon$ is required to cover the residual gas cost needed to strictly exceed $\Delta(\tau, s)$.
We simplify the block capacity optimization problem and assume that $\epsilon$ is sufficiently large such that replacing $\tau'$ with a lower-tip transaction would not increase total revenue through higher gas usage.
Accordingly, the rational block assembler~$\beta_s$ includes~$\tau'$ in $B_s$ to capture its tip, yielding an adversarial expenditure strictly greater than $f^*(\tau)\cdot \Delta(\tau, s) = \mathcal{M}$, where $f^*(\tau) = \ftipgas{\tau} + \fbasegas{\tau'}$ is the per-gas cost required to outrank~$\tau$.
Including~$\tau'$ in $B_s$ exhausts the available gas capacity above $\tau$ in the Fair Ordering: $\tau$ is no longer includable in $B_s$, and by \Cref{def:censorship}, $\tau \notin \textsc{FairFIL}_s$.
Since $\tau \notin \textsc{FairFIL}_s$, block assembler $\beta_{s+1}$ has no obligation to include $\tau$ in $B_{s+1}$. 
The baseline bribe in slot~$s+1$ is therefore sufficient to ensure that $\tau$ is also not included in $B_{s+1}$.
Since $\Bbase_{\tau} < \mathcal{M} + \epsilon$, this per-slot budget suffices for adversary $\mathcal{A}$ to censor $\tau$ for two slots.

\vspace{1mm}
\item[Budget $\mathcal{M}$ does not suffice for two-slot censorship via suppression.]
Assume, for contradiction, that an adversary with budget~$\mathcal{M}$ achieves two-slot censorship solely via the transaction suppression attack.
For $\tau \notin B_{s+1}$ to hold without violating the enforcement of \textsc{FairFIL}, we require $\tau \notin \textsc{FairFIL}_s$; otherwise \Cref{lemma:prop_compliance} implies $\tau \in B_{s+1}$.
For $\tau \notin \textsc{FairFIL}_s$, the cumulative gas of transactions ranking strictly above $\tau$ in the Fair Ordering must exceed $\Delta(\tau, s)$, i.e., the remaining gas capacity after including $\tau$ under Fair Ordering in the block.
With budget $\mathcal{M} = f^*(\tau)\cdot \Delta(\tau, s)$, the adversary can exactly fill the available gas capacity above $\tau$ in $B_s$, making $\tau$ just feasible for inclusion.
By the definition of $\Delta(\tau, s)$, this already exhausts transactions ranking above $\tau$, implying that $\tau$ becomes the next transaction to be considered under Fair Ordering.
Consequently, given budget $\mathcal{M}$, $\tau$ is either included directly in $B_s$ or, by \Cref{lemma:prop_includability,lemma:prop_listing}, appears in $\textsc{FairFIL}_s$ and is therefore included in $B_{s+1}$.
Accordingly, given budget $\mathcal{M}$, adversary $\mathcal{A}$ cannot use the transaction suppression attack to achieve two-slot censorship, contradicting the assumption.\qedhere
\end{description}
\end{proof}
 
\section{Censorship under Prior Censorship Resistance Mechanisms}
\label{appendix:proofs_priorcr}
 
The following four lemmas bound the censorship resistance of FIL, FOCIL, MCP, and AUCIL, surveyed in \Cref{sec:relatedwork}.
The proofs share a common structure: for each mechanism, we describe an adversary~$\mathcal{A}$ with the per-slot bribing budget~$\mathcal{M}$ that prevents the inclusion of an arbitrary target transaction~$\tau$ for an unbounded number of consecutive slots.
By \Cref{def:cr}, this establishes that the mechanism is not $(\mathcal{M},0)$-censorship-resistant for the selected bribing budget.
We do not claim the stated budgets are tight; lower budgets may suffice.
 
\subsection{Forward Inclusion Lists (FIL)}
\label{appendix:proof_fil}
\begin{lemma}
\label{lemma:fil_cr}
With FIL~\cite{EIP7547}, an adversary with per-slot budget
\(
\mathcal{M} = \Bbase_{\tau} + \epsilon
\)
prevents the inclusion of transaction~$\tau$ for an unbounded number of consecutive slots.
\end{lemma}
\begin{proof}
We exhibit adversary $\mathcal{A}_{\textsc{FIL}}$ with per-slot budget $\mathcal{M} = \Bbase_{\tau} + \epsilon$.
In each slot~$s$, $\mathcal{A}_{\textsc{FIL}}$ bribes the block assembler~$\beta_s$ through a single bribing-contract call to (i)~exclude~$\tau$ from~$B_s$ and (ii)~omit~$\tau$ from $\Tau_{\mathrm{req}}$.
For~(i), \Cref{lemma:EthereumCR} establishes that $\Bbase_{\tau} = \ftiptx{\tau} + \cbribe + \epsilon$ suffices to make a rational~$\beta_s$ exclude~$\tau$ from~$B_s$.
For~(ii), FIL provides no per-transaction reward for listing in $\Tau_{\mathrm{req}}$~\cite{EIP7547}, so~$\beta_s$ is indifferent between listing and omitting~$\tau$; an additional margin of~$\epsilon$ makes omission strictly preferable.
Since both deviations are settled through the same bribing-contract call, only one overhead~$\cbribe$ applies, and the total per-slot cost amounts to~$\mathcal{M} = \ftiptx{\tau} + \cbribe + 2\epsilon = \Bbase_{\tau} + \epsilon$.
After the attack, $\tau \notin B_s$ and $\tau \notin \Tau_{\mathrm{req}}$, so no obligation to include $\tau$ in block~$B_{s+1}$ arises.
Hence, $\mathcal{A}_{\textsc{FIL}}$ can repeat the attack in every subsequent slot at the same per-slot cost, and $\tau$ is never included.
\end{proof}
 
\subsection{Fork-Choice Enforced Inclusion Lists (FOCIL)}
\label{appendix:proof_focil}
\begin{lemma}
\label{lemma:focil_cr}
With FOCIL~\cite{thiery2024eip7805}, an adversary with per-slot budget
\(
\mathcal{M} = \kappa\,(\epsilon + \cbribe) + \Bbase_{\tau}
\)
prevents the inclusion of transaction~$\tau$ for an unbounded number of consecutive slots.
\end{lemma}
\begin{proof}
We exhibit adversary $\mathcal{A}_{\textsc{FOCIL}}$ with per-slot budget $\mathcal{M} = \kappa\,(\epsilon + \cbribe) + \Bbase_{\tau}$.
In each slot~$s$, $\mathcal{A}_{\textsc{FOCIL}}$ targets the inclusion committee~$V_s^{\textsc{Focil}}$ of~$\kappa$ includers and the block assembler~$\beta_s$.
For each of the $\kappa$ includers, $\mathcal{A}_{\textsc{FOCIL}}$ pays a bribe of~$\epsilon$ together with the contract overhead~$\cbribe$ to omit~$\tau$ from the includer's partial inclusion list.
Since FOCIL provides no per-transaction reward to includers~\cite{thiery2024eip7805}, the includer is indifferent between listing and omitting~$\tau$, and the margin~$\epsilon$ tips the balance toward omission.
The total committee cost is~$\kappa\,(\epsilon + \cbribe)$.
Once~$\tau$ is omitted from all~$\kappa$ partial lists, the aggregate $\Tau_{\mathrm{req}}$ does not contain~$\tau$, and~$\beta_s$ is not obligated to include~$\tau$ in~$B_s$.
However, a rational block assembler~$\beta_s$ would still include~$\tau$ voluntarily to capture the transaction's tips~$\ftiptx{\tau}$.
$\mathcal{A}_{\textsc{FOCIL}}$ therefore additionally bribes~$\beta_s$ at the Ethereum baseline cost~$\Bbase_{\tau}$ (\Cref{lemma:EthereumCR}) to exclude~$\tau$ from~$B_s$, yielding the stated total per-slot cost~$\mathcal{M}$.
Since FOCIL imposes no obligation on subsequent slots, $\mathcal{A}_{\textsc{FOCIL}}$ can repeat the attack against the subsequently selected committee in every subsequent slot at the same per-slot cost, so~$\tau$ is never included.
\end{proof}
 
\subsection{Multiple Concurrent Proposers (MCP)}
\label{appendix:proof_mcp}
\begin{lemma}
\label{lemma:mcp_cr}
With MCP~\cite{garimidi_concurrent_2025}, an adversary with per-slot budget
\(
\mathcal{M} = \kappa \cdot \Bbase_{\tau}
\)
prevents the inclusion of transaction~$\tau$ for an unbounded number of consecutive slots.
\end{lemma}
\begin{proof}
We exhibit adversary $\mathcal{A}_{\textsc{MCP}}$ with per-slot budget $\mathcal{M} = \kappa \cdot \Bbase_{\tau}$.
MCP replaces the single block assembler with~$\kappa$ concurrent proposers, each producing an independent partial block.
$\mathcal{A}_{\textsc{MCP}}$ targets each of these~$\kappa$ proposers individually, considering the worst case in which every concurrent proposer attempts to include $\tau$ in its partial block.
Including~$\tau$ yields the per-transaction tip~$\ftiptx{\tau}$ to the proposer~\cite{garimidi_concurrent_2025}.
$\mathcal{A}_{\textsc{MCP}}$ compensates each proposer for the forgone tip~$\ftiptx{\tau}$, adds the margin~$\epsilon$ to make exclusion strictly preferable, and pays the contract overhead~$\cbribe$, totaling $\Bbase_{\tau} = \ftiptx{\tau} + \epsilon + \cbribe$ per proposer.
Across all~$\kappa$ proposers, the per-slot cost amounts to~$\mathcal{M} = \kappa \cdot \Bbase_{\tau}$.
Since MCP imposes no obligation propagating beyond the current slot, $\mathcal{A}_{\textsc{MCP}}$ can repeat the attack against the subsequently selected set of~$\kappa$ proposers in every subsequent slot at the same per-slot cost, so~$\tau$ is never included.
\end{proof}
 
\subsection{Auction-Based Inclusion Lists (AUCIL)}
\label{appendix:proof_aucil}
\begin{lemma}
\label{lemma:aucil_cr}
With AUCIL~\cite{aucil}, an adversary with per-slot budget
\(
\mathcal{M} = (\kappa + 1) \cdot \Bbase_{\tau}
\)
prevents the inclusion of transaction~$\tau$ for an unbounded number of consecutive slots.
\end{lemma}
\begin{proof}
We exhibit adversary $\mathcal{A}_{\textsc{AUCIL}}$ with per-slot budget $\mathcal{M} = (\kappa + 1) \cdot \Bbase_{\tau}$.
AUCIL combines a constructing committee of~$\kappa$ includers with a separate block assembler~$\beta_s$ that produces~$B_s$ in the same slot, and rewards each includer for the transactions it contributes~\cite{aucil}.
$\mathcal{A}_{\textsc{AUCIL}}$ adopts the strategy of $\mathcal{A}_{\textsc{MCP}}$ (\Cref{lemma:mcp_cr}) against the~$\kappa$ committee members: each includer is bribed at cost~$\Bbase_{\tau}$ to omit~$\tau$ from its partial inclusion list, totaling~$\kappa \cdot \Bbase_{\tau}$.
Once~$\tau$ is omitted from all~$\kappa$ partial lists, the aggregate $\Tau_{\mathrm{req}}$ does not contain~$\tau$, removing the protocol obligation on~$\beta_s$.
Unlike MCP, however, AUCIL retains a separate block assembler that produces~$B_s$ alongside the committee, and a rational~$\beta_s$ would still include~$\tau$ to capture~$\ftiptx{\tau}$.
$\mathcal{A}_{\textsc{AUCIL}}$ therefore additionally bribes~$\beta_s$ at the Ethereum baseline cost~$\Bbase_{\tau}$ (\Cref{lemma:EthereumCR}) to exclude~$\tau$ from~$B_s$, yielding the stated total per-slot cost $\mathcal{M} = \kappa \cdot \Bbase_{\tau} + \Bbase_{\tau} = (\kappa + 1)\,\Bbase_{\tau}$.
Wadhwa et al.~\cite{aucil} provide a tighter bound on the censorship resistance of AUCIL.
Since AUCIL imposes no obligation propagating beyond the current slot, $\mathcal{A}_{\textsc{AUCIL}}$ can repeat the attack against the subsequently selected committee and assembler in every subsequent slot at the same per-slot cost, so~$\tau$ is never included.
\end{proof}
 
\end{document}